\documentclass[nonacm,sigconf,natbib=true]{acmart}

\copyrightyear{2026}
\acmYear{2026}
\setcopyright{cc}
\setcctype{by}
\acmConference[ICTIR '26]{Proceedings of the 2026 International ACM SIGIR Conference on Innovative Concepts and Theories in Information Retrieval (ICTIR)}{July 25, 2026}{Melbourne, VIC, Australia}
\acmBooktitle{Proceedings of the 2026 International ACM SIGIR Conference on Innovative Concepts and Theories in Information Retrieval (ICTIR) (ICTIR '26), July 25, 2026, Melbourne, VIC, Australia}
\acmDOI{10.1145/3805713.3820417}
\acmISBN{979-8-4007-2600-2/2026/07}

\usepackage{amsthm}
\usepackage{appendix}
\usepackage{array}
\usepackage{balance}
\usepackage{bbm}
\usepackage{blindtext}
\usepackage{comment}
\usepackage{hyperref}
\usepackage{mathtools}
\usepackage{multirow}
\usepackage{relsize}
\usepackage[ruled,vlined]{algorithm2e}

\newcommand*{\E}[2][]{\mathbb{E}_{#1}\left[#2\right]}

\DeclareMathOperator{\ind}{\mathbbm{1}}
\newcommand*{\norm}[1]{\left\lVert#1\right\rVert}
\newcommand*{\abs}[1]{\left| #1 \right|}
\newcommand*{\rb}[1]{\left( #1 \right)}
\newcommand*{\bb}[1]{\left[ #1 \right]}
\newcommand*{\cb}[1]{\left\{ #1 \right\}}

\newcommand*{\maxOn}[2]{\underset{#1}{max}\{#2\}}
\newcommand*{\minOn}[2]{\underset{#1}{min}\{#2\}}
\newcommand*{\argmaxOn}[2]{\underset{#1}{argmax}\cb{#2}}
\newcommand*{\argminOn}[2]{\underset{#1}{argmin}\cb{#2}}
\DeclarePairedDelimiterX\set[1]\lbrace\rbrace{\def\given{\;\delimsize\vert\;}#1}
\newcommand*{\largemath}[1]{\ensuremath{\displaystyle #1}}

\newtheorem{definition}{Definition}
\newtheorem{proposition}{Proposition}  
\newtheorem{corollary}{Corollary}  
\newtheorem{observation}{Observation}  
\newtheorem{lemma}{Lemma}  
\newtheorem{theorem}{Theorem}  
\newtheorem{example}{Example}  
\renewcommand*{\S}{Section~}

\newcommand*{\game}{G}

\newcommand*{\modelName}{publisher game\xspace}
\newcommand*{\modelPlural}{publisher games\xspace}
\newcommand*{\relSt}{relevance strategy\xspace}
\newcommand*{\relPro}{relevance profile\xspace}

\newcommand*{\publishersSet}{N}
\newcommand*{\numPublishers}{n}

\newcommand*{\publisherIndex}{i}
\newcommand*{\publishersPer}{\Pi_\numPublishers}

\newcommand*{\embeddingSpaceBase}{\mathcal{\mathbb{X}}}
\newcommand*{\strategyProfile}{x}

\newcommand*{\embeddingSpaceAux}{X}
\newcommand*{\embeddingSpace}{\embeddingSpaceAux_\publisherIndex}
\newcommand*{\embeddingSpaceSpecific}[1]{\embeddingSpaceAux_{#1}}

\newcommand*{\dimE}{k}
\newcommand*{\strategy}{\strategyProfile_\publisherIndex}
\newcommand*{\strategyTilde}{\tilde{\strategyProfile}_\publisherIndex}
\newcommand*{\strategyTag}{\strategyProfile_{\publisherIndex'}}
\newcommand*{\strategyTagTag}{\strategyProfile_{\publisherIndex''}}
\newcommand*{\strategiesOthers}{\strategyProfile_{-\publisherIndex}}

\newcommand*{\aspect}{a}
\newcommand*{\aspectRep}{\strategyProfile_\aspect}
\newcommand*{\aspectRepSpecific}[1]{\strategyProfile_{\aspect_{#1}}}
\newcommand*{\aspectRepOne}{\aspectRepSpecific{1}}
\newcommand*{\aspectRepTwo}{\aspectRepSpecific{2}}
\newcommand*{\aspectDistribution}{{P^\aspect}}
\newcommand*{\aspectDistributionSupp}{\text{supp}(\aspectDistribution)}
\newcommand*{\aspectDistributionSuppSize}{\abs{\text{supp}(\aspectDistribution)}}
\newcommand*{\aspectE}[1]{\E[\aspect \sim \aspectDistribution]{#1}}
\newcommand*{\aspectEShort}[1]{\E[\aspectDistribution]{#1}}
\newcommand*{\aspectProb}{P_\aspectDistribution(\aspectRep)}
\newcommand*{\aspectProbSpecific}[1]{P_\aspectDistribution(\aspectRepSpecific{#1})}
\newcommand*{\aspectD}{d}
\newcommand*{\baseWeight}{w}
\newcommand*{\baseWeightSet}{W}
\newcommand*{\baseWeightSetAdjusted}{\baseWeightSet^*}
\newcommand*{\aspectWeight}{\baseWeight_\aspect}
\newcommand*{\aspectWeightAdjusted}{\aspectWeight^*}
\newcommand*{\aspectWeightSpecific}[1]{\baseWeight_{\aspect_{#1}}}
\newcommand*{\aspectWeightAdjustedSpecific}[1]{\aspectWeightSpecific{#1}^*}

\newcommand*{\rankFunction}{r}
\newcommand*{\rankFunctionTilde}{\tilde{\rankFunction}}
\newcommand*{\rankIndex}{j}
\newcommand*{\publisherRankLocation}{\rankFunction_\publisherIndex}
\newcommand*{\publisherRankLocationTag}{\rankFunction_{\publisherIndex'}}
\newcommand*{\publisherRankLocationSpecific}[1]{\rankFunction_{#1}}
\newcommand*{\publisherRankTildeLocation}{\rankFunctionTilde_\publisherIndex}
\newcommand*{\publisherRankTildeLocationTag}{\rankFunctionTilde_{\publisherIndex'}}
\newcommand*{\sumLocations}{\sum_{\rankIndex=1}^{\numPublishers}}
\newcommand*{\scoreFunc}{f}
\newcommand*{\partialRank}{L}
\newcommand*{\partialRankSpecific}[1]{\partialRank_{#1}}

\newcommand*{\locFunction}{l}

\newcommand*{\rankLocation}{\locFunction_\rankIndex}

\newcommand*{\rankLocationTag}{\locFunction_{\rankIndex'}}
\newcommand*{\rankLocationSpecific}[1]{\locFunction_{#1}}

\newcommand*{\simFunction}{s}
\newcommand*{\simFunctionGroup}{S}
\newcommand*{\simPublisher}{\simFunction_\publisherIndex^\aspect}
\newcommand*{\simPublisherTag}{\simFunction_{\publisherIndex'}^\aspect}
\newcommand*{\simPublisherTagTag}{\simFunction_{\publisherIndex''}^\aspect}

\newcommand*{\simPublisherSpecific}[1]{\simFunction_{#1}^\aspect}
\newcommand*{\simAllPublishersRank}{\simFunctionGroup^\aspect_\rankFunction}
\newcommand*{\simAllPublishersRankTilde}{\simFunctionGroup^\aspect_{\rankFunctionTilde}}
\newcommand*{\simLocation}{\simFunction_{\rankLocation}^\aspect}

\newcommand*{\simLocationTag}{\simFunction_{\rankLocationTag}^\aspect}
\newcommand*{\simExpected}{\simFunction_\mathbb{E}^\aspect}

\newcommand*{\userUtilityFunction}{v}
\newcommand*{\clickModel}{\userUtilityFunction}
\newcommand*{\publisherUtilityFunction}{u}
\newcommand*{\publisherUtility}{\publisherUtilityFunction_\publisherIndex}

\begin{document}

\title{Stability in Competitive Search with Results Diversification}

\author{Itamar Reinman}
\authornote{Corresponding Author.}
\affiliation{
  \institution{Technion - Israel Institute of Technology}
  \city{Haifa}
  \country{Israel}}
\email{itamarr@campus.technion.ac.il}
\orcid{0009-0003-4749-5004}

\author{Omer Madmon}
\affiliation{
  \institution{Technion - Israel Institute of Technology}
  \city{Haifa}
  \country{Israel}}
\email{omermadmon@campus.technion.ac.il}
\orcid{0009-0001-4009-0368}

\author{Moshe Tennenholtz}
\affiliation{
  \institution{Technion - Israel Institute of Technology}
  \city{Haifa}
  \country{Israel}}
\email{moshet@technion.ac.il}
\orcid{0000-0002-9459-5388}

\author{Oren Kurland}
\affiliation{
  \institution{Technion - Israel Institute of Technology}
  \city{Haifa}
  \country{Israel}}
\email{kurland@technion.ac.il}
\orcid{0000-0002-0669-0431}

\renewcommand{\shortauthors}{Itamar Reinman, Omer Madmon, Moshe Tennenholtz, and Oren Kurland.}

\begin{abstract}
In a competitive search setting, publishers strategically modify their documents in response to induced rankings so as to improve their future ranking. We present a novel game-theoretic analysis of a competitive search setting where search-results diversification is applied. Our analysis reveals an inherent tradeoff between corpus diversity and corpus stability, where the latter corresponds to an equilibrium in a game. We analyze two representative diversification methods and show that stability need not necessarily be reached, leaving the corpus to rapid changes due to ranking incentivized modifications of publishers. We then present a novel approach to devise diversification-based ranking functions that are guaranteed to lead to corpus stability.
\end{abstract}





\maketitle

\section{Introduction}\label{sec:intro}
There is \emph{competitive search} setting, with the Web being a canonical example, where some publishers (document authors) are ranking incentivized \cite{compSearch}. That is, they opt to have their documents highly ranked for specific queries; e.g., those with a commercial intent. As a result, these publishers might respond to induced rankings by modifying their documents. While some of the modifications applied by publishers can be black-hat (e.g., spamming) and hurt the search setting, our focus as in prior work on competitive search \cite{compSearch} is on white-hat content modifications; namely, modifications that do not degrade document quality and more generally do not hurt the search setting.

The publishers' ranking incentives often lead to a ranking competition over queries. There is a line of work that analyzed these ranking competitions using game theory \cite{compSearch}. Specifically, publishers are regarded as players, their actions are the documents they produce, and the ranking function is the mediator.

Reports of game theoretic analysis of ranking competitions (e.g., \cite{basat2017game,raifer2017information}) led to the basic realization that the competitive search setting is fundamentally different than the standard ad hoc setting --- a query used to rank a document corpus without accounting for post-retrieval corpus effects; specifically, in terms of theoretical foundations of information retrieval. A case in point, the probability ranking principle (PRP) \cite{robertson1977probability} was shown to be optimal for ranking documents. It serves as the theoretical underpinning of most retrieval methods. As it turns out, the PRP is sub-optimal in a competitive search setting as it leads to reduced topical diversity in the corpus \cite{basat2017game}. Furthermore, it was shown theoretically and empirically that publishers in ranking competitions tend to mimic content in documents highly ranked in the past \cite{raifer2017information}. This phenomenon of publisher \emph{herding} naturally has unwarranted consequences on the corpus \cite{Goren+al:21a}. Recent work then showed that the herding effect can be ameliorated by applying search-results diversification \cite{mordo2025ameliorating}.

While almost all previous work on modeling competitive search using game theory has focused on ranking solely based on estimated relevance, we present in this paper a novel rigorous game-theoretic modeling of diversification-based ranking. We analyze representative novelty-based and aspect-coverage search-results diversification methods \cite{santos2015search}. Our analysis shows that there is an inherent tradeoff between the extent to which results are diverse and corpus stability. Stability means that the search setting has reached an equilibrium where no publisher has an incentive to deviate from her strategy, namely, the document she produced. It turns out that in quite a few cases, search-results diversification methods do not lead to stability, leaving the corpus to rapid (unwarranted) changes as a result of ranking incentivized modifications.

Given that it is not necessarily guaranteed that existing results-diversification methods will lead to stability, we present a novel approach to devise diversity-based ranking functions with such provided guarantees.

\paragraph{Paper organization}
The paper is structured as follows. In \S\ref{sec:related work} we review related work. \S\ref{sec:ir setting} describes the diversity-based ranking functions we explore. In Sections \ref{sec:model} and \ref{sec:main results lex} we present a game theoretic analysis of the competitive search setting with the ranking functions, including stability (equilibrium) analysis. \S\ref{sec:uir_ranking} presents a framework of devising diversity-based ranking functions that lead to stability. We then conclude and discuss potential future directions in \S\ref{sec:conclusions}. Appendix~\ref{app:proofs} includes the proofs omitted from the main paper. An extended explanation of Example \ref{exa:xmmr no pne} is provided in Appendix~\ref{app:example}. The source code of the numerical analysis we present is available at \url{https://github.com/ireinman/Stability-in-Competitive-Search-with-Results-Diversification}.

\section{Related Work}\label{sec:related work}
Our focus is on competitive search. We note that there is also a line of work on competitive recommendation. \citet{ben2018game} analyzed the competition among content creators, and proposed the Shapley mediator as a recommendation mechanism that provides both stability and fairness guarantees. A sequence of follow-up analyses of content creator competitions under differing assumptions about the incentive structures of creators~\cite{hron2022modeling,jagadeesan2206supply,yao2023howbad} has been conducted, as well as the proposal of novel techniques for various welfare objectives optimization~\cite{mladenov2020optimizing,yao2024rethinking,yao2024userwelfare}.
 
In the competitive search realm, \citet{basat2017game} showed, as noted above, that the PRP is sub-optimal using game theoretic analysis. They defined games of complete information as we do here. They treated documents as discrete uniform distributions over two topics. In contrast, we assume that documents are represented within a continuous embedding space, where discrete or continuous distributions exist over multiple aspects. \citet{basat2017game} assumed a reward only to the most highly ranked document while we assume a reward decreasing in rank. Furthermore, in contrast to \citet{basat2017game}, we analyze search-results diversification methods.

\citet{ben2019convergence} and \citet{madmon2023search,madmon2024convergence} studied the convergence of learning dynamics of strategic publishers to equilibrium. Follow-up works explored publishers' games involving multiple queries~\cite{nachimovsky2024ranking} and proposed novel corpus-enrichment approaches to guarantee ecosystem stability~\cite{nachimovsky2025power}. A recent line of work explored additional economic aspects of the search ecosystem, such as the impact of data sharing~\cite{keinan2025strategic,taitler2025data} and AI overview summarization~\cite{wu2026ai}. In contrast to our work, none of this previous work addressed a search setting with search-results diversification.

The work most related to ours is that of \citet{mordo2025ameliorating} who were the only ones --- to the best of our knowledge --- to use game theory to analyze a competitive search setting where search-results diversification is applied. Our work is different then theirs in several major aspects. \citet{mordo2025ameliorating} assumed a reward to only the two highest-ranked documents, while we account for all documents in a ranked list. Furthermore, they only analyzed a novelty-based diversification method \cite{carbonell1998use}, while we also analyze a representative aspect-coverage diversification approach: xQuAD~\cite{santos2010explicit}. \citet{mordo2025ameliorating} assumed repeated games with incomplete information and analyzed min-max regret equilibria. In contrast, we analyze games with complete information and use the Nash equilibrium concept solution which allows us to focus on the stability or lack-thereof of the corpus and its tradeoff with corpus diversity. Furthermore, we present a novel approach to define search-results diversification methods that are guaranteed to lead to corpus stability (i.e., equilibrium) while \citet{mordo2025ameliorating} only analyzed a single existing diversification method.

As is standard practice in work on applying game theoretic analysis to competitive search \cite{compSearch}, we focus on a setting where a publisher competes for a single query. \citet{nachimovsky2024ranking} analyzed a setting where a publisher competes for a few queries representing the same information need. The finding was in line with reports of competing for a single query \cite{raifer2017information}: publishers mimic content in documents that were previously highly ranked. In contrast to our work, \citet{nachimovsky2024ranking} did not analyze a setting with search-results diversification. We leave the analysis of competitions for multiple queries with diversity-based ranking for future work.

\section{Diversity-Based Ranking Model}\label{sec:ir setting}

We begin by describing a formal framework for diversity-based ranking which will serve as a basis for the analysis to follow. The framework addresses the task of ranking publishers' documents in a corpus in response to a query with the goal of accounting for both relevance and diversification of search-results. Unless otherwise specified, we assume a fixed corpus of documents and a query. Table~\ref{tab:notations} summarizes all the notations in our paper.

\begin{table*}[!t]
\caption{\label{tab:notations} Notational conventions.}
\centering
\resizebox{\textwidth}{!}{
\begin{tabular}{|>{\largemath}l|p{5.5cm}|>{\largemath}l|p{5.5cm}|}
    \hline
    \textbf{Notation} & \textbf{Description} & \textbf{Notation} & \textbf{Description} \\
    \hline
    $\publishersSet$ & Set of documents: $\{1, 2, \ldots, \numPublishers\}$ & $\rankLocation$ & The index of the document at position $\rankIndex$ \\
    \hline
    $\publishersPer$ & Set of all permutations of $\publishersSet$ & $\simFunction$ & Similarity function on embedding space \\
    \hline
    $\embeddingSpaceBase$ & Embedding space: $[0,1]^{\dimE}$ & $\simPublisher$ & Similarity between $\strategy$ and $\aspectRep$ \\
    \hline
    $\strategy$ & The embedding of document $\publisherIndex$ & $\simAllPublishersRank$ & Similarity scores ordered by $\rankFunction$ \\
    \hline
    $\strategyProfile$ & The tuple of the documents' embeddings & $\scoreFunc$ & Score function \\
    \hline
    $\aspect$ & A query aspect & $\partialRank$ & Documents previously ranked \\
    \hline
    $\aspectRep$ & The embedding of query aspect $\aspect$ & $\lambda$ & Linear weight coefficient \\
    \hline
    $\aspectDistribution$ & Aspect distribution & $\clickModel$ & User utility function \\
    \hline
    $\aspectDistributionSupp$ & Aspect distribution's support & $\publisherUtilityFunction$ & Publisher's utility function \\
    \hline
    $\rankFunction$ & Ranking function & $\publisherIndex$ & Document/publisher index \\
    \hline
    $\publisherRankLocation$ & Rank position of document $\publisherIndex$ & $\rankIndex$ & Position/location index \\
    \hline
\end{tabular}
}
\end{table*}

\paragraph{Query Aspects}
As is standard in work on search-results diversification, we assume that the query manifests a few aspects~\cite{santos2015search}; $\aspect$ denotes an aspect. We do not subscribe to a specific definition of an aspect. It is convenient, however, to think of an aspect as a facet of the underlying information need as is standard practice~\cite{santos2015search}; e.g., for the query ``iphone'', it could be ``price'', ``color'', etc.

We assume that an aspect distribution $\aspectDistribution$ was induced from the query using \emph{some} approach; e.g., based on clustering or topic modeling~\cite{santos2015search}. The distribution can be discrete or continuous. In the spirit of the probability ranking principle (PRP)~\cite{robertson1977probability}, we assume that this is the distribution that bests describes the underlying aspects of the query (information need) given the information available to the search system\footnote{In what follows, our stochastic treatment of aspects is in line with the PRP's stochastic treatment of relevance~\cite{robertson1977probability}. That is, while a document is either relevant to a query or not the PRP assumes a relevance probability which conceptually (and implicitly) amounts to the probability that the ``average user'' will deem the document relevant. The aspect distribution we refer to here can be thought of as modeling the ``belief'' of the average user about the actual query aspects.  Using a specific aspect-distribution induction method results in an estimate for the distribution the same way that a specific relevance estimation method in the PRP results in an estimate for relevance.}.

\paragraph{Documents and embedding representations}
We assume a dense retrieval model~\cite{izacard2021unsupervised,wang2022text,bruch2024foundations}, where both documents and the query aspects are represented in a $\dimE$-dimensional embedding space $\embeddingSpaceBase = [0,1]^\dimE$. 
In what follows, we refer to a document and its embedding representation interchangeably\footnote{Following prior work~\cite{goren2018ranking,tennenholtz2024demystifying,madmon2023search,madmon2024convergence}, we assume that every possible representation in the embedding space corresponds to a valid document, that is, the embedding function is surjective.}.  We use $\publishersSet \coloneqq \{1, 2, \ldots, \numPublishers \}$ to denote a finite set of document indices, where $\strategy \in \embeddingSpaceBase$ is the embedding representation of document $\publisherIndex \in \publishersSet$.  We denote the set of all permutations of $\publishersSet$ by $\publishersPer$.  We denote by $\strategyProfile \in \embeddingSpaceBase^\numPublishers$ the tuple consisting of all document representations, and we write $\strategyProfile_{-\publisherIndex}$ to denote the tuple consisting of all entries of $\strategyProfile$ except for the $\publisherIndex$'s entry. Query aspects are also represented in the same embedding space, with $\aspectRep \in \embeddingSpaceBase$ denoting the embedding representation of aspect $\aspect$.

\paragraph{Ranking functions}
A ranking function, denoted $\rankFunction$, takes as an input a set of document embeddings $\cb{\strategy}_{\publisherIndex \in \publishersSet}$ and an induced aspect distribution $\aspectDistribution$, and returns a permutation (ranked list) of the documents. Following standard practice in work on dense retrieval~\cite{izacard2021unsupervised,wang2022text,bruch2024foundations}, we assume that the ranking function utilizes a \emph{similarity function} $\simFunction: \embeddingSpaceBase \times \embeddingSpaceBase \to [0,1]$ that estimates the level of relevance of a document $\strategy$ to a representation $\aspectRep$.  Throughout the analysis presented below, we use the following similarity function, which is rank-equivalent to the negative Euclidean distance:
\begin{equation*}
    \simFunction(\strategy, \aspectRep) = 1 - \frac{1}{\dimE}\norm{\strategy-\aspectRep}^2.
\end{equation*}
For any document-aspect pair, $\simPublisher = \simFunction(\strategy, \aspectRep)$ is an estimate for the relevance to $\aspect$.\footnote{Another common similarity function is the cosine similarity. If we assume that the vectors in the embedding space are normalized, then the cosine similarity function is rank-equivalent to the negative Euclidean norm.}
We denote the position of document $\publisherIndex$ by $\publisherRankLocation$ and the original index of the document that is ranked at position $\rankIndex$ by $\rankLocation$. In addition, we use $\simAllPublishersRank = {(\simLocation)}_{\rankIndex \in \bb{\numPublishers}}$ to denote the list of \emph{similarity values} of all the documents in the corpus to aspect $\aspect$ when they are ordered by the ranking.

\subsection{Diversity}\label{sub:rank diversity examples}
We analyze diversity-based ranking functions which operate in an iterative manner:  at each position, every unranked document $i$ is assigned a retrieval score, $$\scoreFunc(\strategy, \aspectDistribution, \partialRank),$$ where $\strategy$ is the representation of document $\publisherIndex$, $\aspectDistribution$ is the distribution over query aspects, and $\partialRank$ is the list of documents already ranked. The document assigned the highest score is selected.
The recursive procedure described in Algorithm~\ref{alg:score_func} is then used to induce a ranking. We use $a \mathbin{\|} b$ to denote the concatenation of element $b$ to the tuple $a$.

\begin{algorithm}[ht]
\KwIn{Documents $\cb{\strategy}_{\publisherIndex \in \publishersSet}$, aspect distribution $\aspectDistribution$}
$\partialRankSpecific{0} \gets \varnothing$\\
\For{$\rankIndex=1$ \KwTo $\numPublishers$}{
        $\rankLocation \gets \arg\max_{\publisherIndex \notin \partialRankSpecific{\rankIndex-1}} \scoreFunc(\strategy,\aspectDistribution,\partialRankSpecific{\rankIndex-1})$\\
        $\partialRankSpecific{\rankIndex} \gets \partialRankSpecific{\rankIndex-1} \mathbin{\|} \rankLocation$
    }
\KwOut{A permutation $(\rankLocation)_{\rankIndex\in [\numPublishers]} \in \publishersPer$}
\caption{Ranking procedure using the retrieval score function $\scoreFunc$.}
\label{alg:score_func}
\end{algorithm}

We note that for the ranking procedure to be well-defined, a tie-breaking rule should be specified, as there might be retrieval-scores ties. In the context of deterministic functions, we can, without loss of generality, assume a lexicographical tie-breaking, which will be assumed throughout the paper. Lexicographical tie-breaking can be seen as capturing external factors used by the ranker, such as PageRank scores.

In our theoretical analysis, we consider two retrieval-score functions, and the ranking procedures which use them in Algorithm \ref{alg:score_func}: xQuAD~\cite{santos2010explicit}, and xMMR~\cite{santos2012role}. xQuAD and xMMR represent two widely used approaches for search-results diversification as we discuss below.

\paragraph{The xQuAD retrieval method}
The xQuAD score function~\cite{santos2010explicit} aims to balance
document relevance and \emph{aspect-coverage}, controlled by a
hyperparameter $\lambda \in (0,1)$. The relevance term is the expected similarity with the query aspects, while the aspect-coverage term is a weighted average of the similarity scores, where weights are determined by the aspect distribution and the similarity scores of the previously ranked documents.  Intuitively, remaining documents that are relevant to aspects not yet covered by higher-ranked documents receive higher scores.

\begin{definition}[xQuAD]\label{def:xquad score}
    The xQuAD score of document $\publisherIndex$ is:
    \begin{equation*}
        \scoreFunc_{xQuAD}(\strategy, \aspectDistribution, \partialRank) = \lambda\aspectEShort{\simPublisher} + (1 - \lambda)\aspectEShort{\simPublisher \cdot \underset{\strategyTag \in \partialRank}{\Pi} \big(1 - \simPublisherTag\big)}
    .\end{equation*}
\end{definition}

Note that the second term (aspect-coverage) down-weights aspects that have already been well covered by previously ranked documents, thereby promoting diversity.

\paragraph{The xMMR retrieval method}
The xMMR score function~\cite{santos2012role}, which is based on the MMR framework~\cite{carbonell1998use}, encourages diversity by penalizing redundancy rather than explicitly incentivizing aspect-coverage as in xQuAD. To measure redundancy between two documents, each document is represented as a vector of similarity scores with respect to all aspects. Redundancy between two documents is calculated using a similarity metric. The redundancy of a candidate document with respect to a set of already-ranked documents is obtained by taking a maximum aggregation over these pairwise redundancy terms. 
The final xMMR score balances relevance and a redundancy penalty via a trade-off parameter $\lambda \in (0,1)$. 

\begin{definition}[xMMR]\label{def:xmmr score}
    The xMMR score of document $\publisherIndex$ is:
    \begin{align*}
        \scoreFunc_{xMMR}(\strategy, \aspectDistribution, \partialRank)
        & = \lambda \aspectEShort{\simPublisher} - (1 - \lambda) \maxOn{\strategyTag \in \partialRank}{\aspectEShort{1 - (\simPublisher - \simPublisherTag)^2}}
        \\ & \overset{\text{Rank}}{\equiv} 
        \lambda \aspectEShort{\simPublisher} + 
        (1 - \lambda) \minOn{\strategyTag \in \partialRank}{\aspectEShort{(\simPublisher - \simPublisherTag)^2}},
    \end{align*}
    where $\lambda \in (0,1)$ is hyperparameter, and we assume $\max \emptyset = 0$ and $\min \emptyset = 1$.
\end{definition} 

Throughout the paper, we will use the latter form, in which the redundancy penalty component is replaced with a positive \emph{novelty} component.

\subsection{User Utility Functions}\label{sub:user_utility}
In this work, we adopt position-based user utility functions to capture the gain of users based on the ranked search-results. In a position-based model, the user utility resulting from each retrieved document depends on the document position, the document's estimated relevance to the query aspect, and the estimated aspect-based relevance of the other documents. For simplicity, we focus on user utility functions which are bounded in $[0,1]$. Formally, the \emph{user utility} is a function

\begin{equation*}
    \clickModel: \publishersSet \times {[0, 1]}^\numPublishers \rightarrow [0,1],
\end{equation*}

where $\clickModel(\rankIndex, \simAllPublishersRank)$ represents the utility that a user is gaining from a document ranked at position $\rankIndex$, given that the user's intent corresponds to aspect $\aspect$, and the ordered similarity scores are $\simAllPublishersRank$. An important property of a user utility function is monotonicity:

\begin{definition}[Monotone user utility]
    A user utility function $\clickModel$ is monotone if for every document tuple $\strategyProfile$, ranking $\rb{\publisherRankLocation}_{\publisherIndex \in \publishersSet}$ and pair of documents $\publisherIndex, \publisherIndex'$ such that $\publisherRankLocation > \publisherRankLocationTag$, switching the documents' positions increases the utility gained from document $\publisherIndex$.
\end{definition}

This definition captures the intuitive behavior of search engine users, where visibility and attention are finite resources that diminish as one moves further down a list, which reduces the contribution of lower-ranked documents. Examples of modeling this user behavior model can be seen in the NDCG~\cite{jarvelin2002cumulated} and ERR~\cite{chapelle2009expected} evaluation measures.

In our analysis, we focus on two user utility functions which represent
two approaches in search-results diversification: an
aspect-coverage-based user utility function and a novelty-based user
utility function.\footnote{Note that when examining different options for the user utility function, it is not possible to choose functions in which there are separate components such that one or more of them is independent of the ranking, as this will lead to unrealistic phenomena. For example, if we assume that the function is of the following structure: $\simPublisher + (\text{second-expression})$, then we can see that when the number of publishers $\numPublishers$ tends to infinity, it is possible to reach $\sumLocations\clickModel(\rankIndex, \simAllPublishersRank) \geq \sumLocations \simLocation$ that is not bounded, which is an unrealistic phenomenon - the utility of a user from a ranked list is clearly bounded from above.}

\begin{equation*}
    \clickModel^{coverage}(\rankIndex, \simAllPublishersRank) = \simLocation \cdot \underset{\rankIndex' < \rankIndex}{\Pi} \big(1 - \simLocationTag\big).
\end{equation*}

\begin{equation*}
    \clickModel^{novelty}(\rankIndex, \simAllPublishersRank) = \simLocation \cdot \minOn{\rankIndex' < \rankIndex}{{(\simLocation - \simLocationTag)}^2}.
\end{equation*}

\section{Diversity-Based Ranking in Competitive Search}\label{sec:model}

In a competitive search setting~\cite{compSearch}, publishers are ranking incentivized: they often modify their documents in response to induced rankings so as to improve future ranking.  
We next introduce a \emph{game-theoretic} model, in which documents are produced strategically by publishers, who aim to maximize their utility. Formally, a \emph{\modelName} is defined by a tuple $\game = \rb{\publishersSet, \aspectDistribution, \dimE, \rankFunction, \clickModel}$, where $\publishersSet \coloneqq \{1, 2, \ldots, \numPublishers \}$ is a set of publishers (players), $\aspectDistribution$ is an aspect distribution, $\dimE$ is the dimension of $\embeddingSpaceBase$, $\rankFunction$ is a ranking function, and $\clickModel$ is a user utility function. In this game, each publisher's strategy space $\embeddingSpace$ corresponds to the embedding space, namely $\embeddingSpace = \embeddingSpaceBase, \; \forall \publisherIndex \in \publishersSet$. From now on, when we use the term publisher $\publisherIndex$, we refer to the publisher who publishes document $\publisherIndex$.
As in past work on competitive search \cite{basat2017game}, we assume an alignment between publisher utility and user utility. More specifically, given a strategy profile $\strategyProfile$ (i.e., a vector specifying each publisher's strategy), the utility of publisher $\publisherIndex$ is the expected user utility, where the expectation is taken with respect to the aspect distribution:

\begin{equation*}
    \publisherUtility(\strategyProfile) = \aspectE{\clickModel(\publisherRankLocation, \simAllPublishersRank)}.
\end{equation*}

It is important to point out that the choice of the ranking function $\rankFunction$ significantly affects the publishers' strategic behavior, as it directly shapes their incentive structure.

The games we consider are of \emph{complete information}. Specifically, publishers are assumed to know the embedding approach used by the ranker, the similarity function used to compare aspects and queries with documents, and the method used to induce query aspects; i.e., the ranking function. This assumption is conceptually similar to that taken in work on white box adversarial attacks on ranking functions~\cite{Samadi+al:21a}. We note that previous work on competitive search~\cite{basat2017game} also analyzed publishers' games of complete information, but made a stronger assumption than that we take here; namely, that the true relevance status of a document is known to all players (publishers).

Assuming that publishers do not know the ranking function leads to \emph{Bayesian games} with incomplete information whose treatment is outside the scope of this paper. We hasten to point out, however, that our goal is analyzing (diversity-based) ranking approaches, and more specifically, their effect on publishers' strategies and the stability of the resultant corpus as we discuss next. Hence, neutralizing the effect of the quality of estimates/beliefs used by publishers with respect to the ranking approach allows for a rigorous analysis and understanding of the actual effect of the ranking functions, and more specifically, their underlying principles, on the strategic behavior of publishers and the resultant effects on the corpus. 

In our game theoretic framework, we evaluate ranking functions in terms of stability and diversity. The following sections discusses these concepts in detail.

\subsection{Stability of the Retrieval Setting}

An important notion in game theory is the \emph{Nash equilibrium}, which represents a stable state, namely, a strategy profile from which no player has an incentive to modify its strategy, given that all other players' strategies remain fixed\footnote{Following prior work on   game-theoretic modeling of information retrieval, we adopt the notion of Nash equilibrium in \emph{pure} strategies, and do not consider \emph{mixed strategies} (i.e., allowing for a non-deterministic choice of strategies). See \S7 of~\citet{madmon2023search} for an in-depth discussion on this modeling choice.}. Formally:

\begin{definition}[Nash equilibrium] \label{def:pne}
    A strategy profile $\strategyProfile \in \prod_{\publisherIndex=1}^{\numPublishers} \embeddingSpace$ is a Nash equilibrium (NE) if no player has a profitable unilateral deviation, i.e.,
    \[
        \publisherUtility(\strategyProfile) \geq \publisherUtility(\strategy', \strategiesOthers)
        \quad \text{for all } \publisherIndex \in \publishersSet \text{ and all } \strategy' \in \embeddingSpace.
    \]
\end{definition}

In cases where the system reaches equilibrium, we can examine whether the equilibrium is \emph{diverse}; i.e., if there is variation among the documents. As discussed before, \citet{mordo2025ameliorating} showed that the "mimicking-the-winner" strategy presented in~\citet{raifer2017information} can be ameliorated by using diversity-based ranking methods. Consequently, if a diversity-based retrieval method fails to prevent (significantly) reduced diversity in an equilibrium, its suitability for competitive setting should be called into question.

The strongest notion of stability arises when each player has a single strategy that is optimal \emph{regardless} of the other players' strategies. This idea is captured by the notion of \emph{dominant strategies}:

\begin{definition}[Dominant strategy]\label{def:dominant strategy}
    A strategy $\strategy \in \embeddingSpace$ is a dominant strategy if for any $ \strategy' \in \embeddingSpace$, and for any $\strategiesOthers \in \prod_{\publisherIndex' \neq \publisherIndex} \embeddingSpaceSpecific{\publisherIndex'}$, it holds that $\publisherUtility(\strategy, \strategiesOthers) > \publisherUtility(\strategy', \strategiesOthers)$.
\end{definition}

Clearly, any player in any game can only have at most one dominant strategy. Arguably, a game in which each player has a dominant strategy is degenerate, in the sense that rational players will always play their dominant strategies. It is straightforward to prove the following well-known proposition:

\begin{proposition}\label{pro:dominant implies pne}
    If every player has a dominant strategy, the resulting strategy profile is a Nash equilibrium, and is the unique Nash equilibrium of the game.
\end{proposition}

\subsection{Games Induced by Ranking Functions}\label{sub:games examples}

We now instantiate the \modelPlural induced by each of the two ranking functions discussed in \S\ref{sec:ir setting}. (Recall that, unless stated otherwise, we adopt the lexicographical tie-breaking rule.) To instantiate the induced game for each ranking function, we utilize the user utility function representing the same approach, as presented in \S\ref{sub:user_utility}. The publisher utilities in each of the two families of induced games can be compactly written as:

\begin{equation*}
    \publisherUtility^{xQuAD}(\strategyProfile)
    = \aspectEShort{\clickModel_{coverage}(\publisherRankLocation, \simAllPublishersRank)}
    = \aspectEShort{\simPublisher \cdot \underset{\rankIndex' < \publisherRankLocation}{\Pi} \big(1 - \simLocationTag\big)},
\end{equation*}
\begin{equation*}
    \publisherUtility^{xMMR}(\strategyProfile)
    = \aspectEShort{\clickModel_{novelty}(\publisherRankLocation, \simAllPublishersRank)}
    = \aspectEShort{\simPublisher \cdot \minOn{\rankIndex' < \publisherRankLocation}{(\simPublisher - \simLocationTag)^2}}.
\end{equation*}

The remainder of the paper is concerned with the analysis of the \modelPlural induced by the various ranking functions in terms of stability and diversity. From now on, unless explicitly stated, when we use the term \emph{utility}, we refer to the publisher utility function.

\begin{table*}[!t]
\caption{\label{tab:summary} Summary of the theoretical results presented in Sections~\ref{sec:main results lex} and \ref{sec:uir_ranking}.}
\centering
\resizebox{\textwidth}{!}{
\begin{tabular}{|l|l|c|c|c|l|}
    \hline
    \textbf{Ranking Function} 
    & \textbf{Aspect Distribution} 
    & \textbf{\# Publishers} 
    & \textbf{Equilibrium Exists?} 
    & \textbf{Diverse Equilibrium?} 
    & \textbf{Formal Result} \\
    \hline
    \multirow{2}{*}{xQuAD}
    & Symmetric 
    & Any 
    & Yes (unique) 
    & No 
    & Theorem~\ref{the:xquad lex dominant} \\
    \cline{2-6}
    & Asymmetric 
    & $\ge 3$ 
    & Not guaranteed 
    & -- 
    & Observation~\ref{obs:xquad non-sym no pne} \\
    \hline
    \multirow{2}{*}{xMMR}
    & Any
    & $2$ 
    & Yes 
    & Yes 
    & Corollary~\ref{cor:xmmr trivial pne} \\
    \cline{2-6}
    & Any
    & $\ge 3$ 
    & Not guaranteed 
    & Yes (if exists)
    & Observation~\ref{obs:xmmr copycat zero} and Example~\ref{exa:xmmr no pne} \\
    \hline
    Utility-induced 
    & Any 
    & Any 
    & Yes 
    & Not guaranteed
    & Theorem~\ref{the:direct pne} \\
    \hline
    \multirow{2}{*}{UIR--xQuAD}
    & Symmetric
    & Any 
    & Yes (unique) 
    & No 
    & \multirow{2}{*}{Lemma~\ref{lem:xquad uir unique pne}} \\
    \cline{2-5}
    & Asymmetric (two aspects)
    & Any 
    & Yes (unique) 
    & Yes 
    & \\
    \hline
    UIR--xMMR 
    & Any 
    & Any 
    & Yes 
    & Yes 
    & Lemma~\ref{lem:xmmr uir pne} \\
    \hline
\end{tabular}
}
\end{table*}

\section{Stability and Diversity Analysis}\label{sec:main results lex}
We next present a game-theoretic analysis of our competitive search setting under diversity-based ranking functions. In particular, we study the \modelPlural induced by the xQuAD, and the xMMR ranking functions in terms of stability and diversity. Table~\ref{tab:summary} summarizes the theoretical results presented in this section and in \S\ref{sec:uir_ranking}.
A strategy that will play a central role in our analysis is the \emph{\relSt}, defined as the strategy that matches the expected aspect representation (with respect to a given aspect distribution).

\begin{definition}[Relevance strategy]\label{def:relevance strategy}
    The \relSt is $\aspectE{\aspectRep}$, and the \relPro is a strategy profile in which all publishers play the \relSt.
\end{definition}

Note that the \relSt uniquely maximizes the expected similarity score $\aspectEShort{\simPublisher}$, hence it serves as a natural benchmark for studying publishers' strategic behavior.

\begin{proposition}\label{pro:sim maximizer}
    $\aspectEShort{\aspectRep}$ is the unique maximizer of $\; \aspectEShort{\simPublisher}$.
\end{proposition}

The proof is simple and relies on the fact that maximizing the expected similarity score is equivalent to minimizing the MSE error, whose unique minimizer is the distribution mean.

In the next subsections, we study the effect of the two diversity-based ranking functions, xQuAD and xMMR, on stability and diversity in the induced competitive environment.

\subsection{Stability and Diversity under xQuAD}\label{sub:xquad}

We now turn to analyze xQuAD \modelPlural, in which the ranking function aims to promote diversity by explicitly incentivizing aspect coverage. Surprisingly, we show that for a broad class of aspect distributions, xQuAD ranking induces a degenerate game, in which the \relSt is a dominant strategy for all publishers. What this means is that if publishers act rationally, they will publish exactly the same document. This result is rather surprising, as in xQuAD games both the score function and the utility function supposedly prioritize coverage.

We begin by defining the notion of \emph{symmetric aspect distributions}. We say that $\aspectDistribution$ is symmetric if, for any distance $\aspectD > 0$, the expected value of $\aspectRep$ conditioned on being at that distance from the mean equals the mean itself. Formally:

\begin{definition}[Symmetric aspect distribution]\label{def:sym dist}
    An aspect distribution $\aspectDistribution$ is symmetric if for every $\aspectD \geq 0$, 
    \[
    \aspectEShort{\aspectRep \mid \norm{\aspectEShort{\aspectRep} - \aspectRep} = \aspectD } = \aspectEShort{\aspectRep}.
    \]
\end{definition}

The class of symmetric distributions contains, for instance, spherical distributions, quadratic uniform distributions, and truncated multivariate normal distributions provided that the truncation is symmetric around the mean.

We now show that in xQuAD games with symmetric aspect distribution, if all publishers that precede publisher $\publisherIndex$ in the ranking follow the \relSt, the unique maximizer of the user utility from the document at location $\publisherIndex$ is also the \relSt as well.

\begin{lemma}\label{lem:xquad recursive maximizer}
    Let $\aspectDistribution$ be a symmetric aspect distribution, and let $\publishersSet' \subseteq \publishersSet$ be a subset of publishers. If $\; \forall \publisherIndex' \in \publishersSet',$ $\strategyTag = \aspectEShort{\aspectRep}$, then
    \begin{equation*}
        \aspectEShort{\aspectRep} = \argmaxOn{\strategy}{\aspectEShort{\simPublisher \underset{\publisherIndex' \in \publishersSet'}{\Pi} \big(1 - \simPublisherTag\big)}}
    \end{equation*}
\end{lemma}

The proof relies on the fact that for symmetric aspect distributions, all aspects with the same distance from the expectation are being uniformly reweighted in the aspect-coverage term, preserving the optimality of the \relSt. A full proof of the lemma can be found in Appendix~\ref{proof lem:xquad recursive maximizer}. Using this lemma, we will show that playing the \relSt is an optimal choice in terms of ranking, and that if a publisher plays a strategy other than the \relSt, their ranking will necessarily be worse than the ranking of a publisher who did follow this strategy.

\begin{lemma}\label{lem:xquad lex relevance better rank}
     Let $\game$ be an xQuAD game with a symmetric aspect distribution, and let $\publisherIndex, \publisherIndex'$ be a pair of publishers. If $\strategy = \aspectEShort{\aspectRep}$ and $\strategyTag \neq \aspectEShort{\aspectRep}$, then, $\publisherRankLocation < \publisherRankLocationTag$.
\end{lemma}

The proof of the lemma relies on Proposition~\ref{pro:sim maximizer} and Lemma~\ref{lem:xquad recursive maximizer}. A full proof of the lemma can be found in Appendix~\ref{proof lem:xquad lex relevance better rank}. By definition, the expected user utility is also the utility of a publisher. Having established that the \relSt leads to the best possible position while at the same time it maximizes the user utility from this position (since those who are ranked above necessarily play the \relSt as well), we prove the following theorem:

\begin{theorem}\label{the:xquad lex dominant}
    In any xQuAD \modelName with a symmetric aspect distribution, the \relSt is a dominant strategy for all the publishers.
\end{theorem}

\begin{proof}
    Let $\game$ be a xQuAD \modelName with a symmetric aspect distribution. Fix some publisher $\publisherIndex \in \publishersSet$. We will show that the \relSt of playing $\aspectEShort{\aspectRep}$ is a dominant strategy. 
    
    Let $\strategiesOthers \in \prod_{\publisherIndex' \neq \publisherIndex} \embeddingSpaceSpecific{\publisherIndex'}$ be some strategies of all players except player $\publisherIndex$, $\strategy^1 = \aspectEShort{\aspectRep}$ and $\strategy^2 \neq \aspectEShort{\aspectRep}$.

    According to Lemma~\ref{lem:xquad lex relevance better rank}, every publisher positioned above publisher $\publisherIndex$ played $\aspectEShort{\aspectRep}$. In addition, according to the lemma, if publisher $\publisherIndex$ will deviate to $\strategy^2$, her position will not improve, as she can not be ranked above a publisher who plays $\aspectEShort{\aspectRep}$. We denote the rank of publisher $\publisherIndex$ when playing $\strategy^1$ by $\publisherRankLocation^1$  and the of publisher $\publisherIndex$ when playing $\strategy^2$ by $\publisherRankLocation^2$ and we get that $\publisherRankLocation^1 \leq \publisherRankLocation^2$.
    Therefore:
    \begin{align*}
        \publisherUtility(\strategy^2, \strategiesOthers) & = \aspectEShort{\simFunction(\strategy^2, \aspectRep) \cdot \underset{\rankIndex' < \publisherRankLocation^2}{\Pi} \big(1 - \simLocationTag\big)}
        \\ & \overset{(1)}{\leq} 
        \aspectEShort{\simFunction(\strategy^2, \aspectRep) \cdot \underset{\rankIndex' < \publisherRankLocation^1}{\Pi} \big(1 - \simLocationTag\big)}
        \\ & \overset{(2)}{<} 
        \aspectEShort{\simFunction(\strategy^1, \aspectRep) \cdot \underset{\rankIndex' < \publisherRankLocation^1}{\Pi} \big(1 - \simLocationTag\big)} = \publisherUtility(\strategy^1, \strategiesOthers)
    ,\end{align*}
    where transition~(1) is based on that $\publisherRankLocation^1 \leq \publisherRankLocation^2$ and that all the components are non-negative, and transition~(2) is Lemma~\ref{lem:xquad recursive maximizer}.

    As a result, by definition, the \relSt of playing $\aspectEShort{\aspectRep}$ is a dominant strategy for every publisher $\publisherIndex$.
\end{proof}

Theorem~\ref{the:xquad lex dominant} implies that under aspect
symmetry, xQuAD-based games are degenerated in the sense that rational
publishers will all choose the same document representation, despite
the explicit diversity objective imposed by the xQuAD ranking
function. It is important to note that equilibrium represents the state-of-affairs in the long run. More specifically, Theorem~\ref{the:xquad lex dominant} implies an {\em herding} effect which was already observed in work on analyzing ranking games where ranking depends solely on relevance estimates (and not diversity) \cite{raifer2017information}.

One possible explanation to the findings discussed above is that xQuAD incentivizes coverage rather than novelty. Intuitively, if the documents ranked first and second are identical, then a reasonable user will not gain any value from the second-ranked document. This consideration is not reflected in the aspect-coverage-based user utility function but is captured by the novelty-based user utility function discussed in \S\ref{sub:xmmr}.\footnote{Note that in the two identical top-ranked documents illustration, in the novelty-based user utility function the term $\minOn{\rankIndex' < \publisherRankLocation}{(\simPublisher - \simLocationTag)^2}$ becomes zero, which heavily punishes mimicking documents of previously ranked publishers.}

We conclude the analysis of the xQuAD-induced games by demonstrating that under a non-symmetric aspect distribution, the previous results break and an equilibrium may not exist. A full proof of the observation can be found in Appendix~\ref{proof obs:xquad non-sym no pne}.

\begin{observation}\label{obs:xquad non-sym no pne}
    Let $\game$ be an xQuAD three-publisher game with $\dimE = 1, \lambda = 0.5$ and the following non-symmetric aspect distribution:
    \begin{equation*}
        \aspectProb = 0.75 \cdot \ind(\aspectRep = 0) + 0.25 \cdot \ind(\aspectRep = 1).
    \end{equation*} 
    Then, $\game$ possesses no Nash equilibria.
\end{observation}

In \S\ref{sec:uir_ranking} we will show that the xQuAD ranking function can be adjusted to guarantee both stability and diversity under non-symmetric aspect distributions.

\subsection{Stability and Diversity under xMMR}\label{sub:xmmr}

As discussed in \S\ref{sub:rank diversity examples}, while xQuAD
aims to promote diversity via aspect coverage, xMMR directly
encourages novelty. In \S\ref{sub:xquad} we showed that under
aspect distribution symmetry, the coverage-driven approach used in
xQuAD ranking fails to result in diversity in the publishers'
equilibrium; in the induced game  all publishers adopt the same
strategy, essentially trading diversity for stability. In this
section, we show that xMMR ranking results in a fundamentally
different outcome. As a first step, we note that the \relPro is not
even a Nash equilibrium in xMMR games with symmetric aspect
distributions.

\begin{observation}\label{obs:xmmr copycat zero}
    Let $\game$ be an xMMR \modelName, and let $\publisherIndex > \publisherIndex'$ be two publishers. Then:
    \begin{enumerate}
        \item If $\strategy = \strategyTag$, then the utility of publisher $\publisherIndex$ is $0$.
        \item If a strategy profile $\strategyProfile$ is a Nash equilibrium in $\game$, then $\strategy \neq \strategyTag$.
\end{enumerate}
\end{observation}

The intuition is simple. Under the xMMR ranking, any publisher whose document is identical to at least one previously ranked document results in a zero novelty score. When all documents are identical, as is the case in the \relPro, it is straightforward to see that only publisher 1 (i.e., the first in lexicographic order) receives strictly positive utility (under lexicographical tie-breaking), while any other publisher can deviate and benefit a strictly positive gain in utility. A full proof of the observation can be found in Appendix~\ref{proof obs:xmmr copycat zero}.

While Observation~\ref{obs:xmmr copycat zero} rules out the possibility that the \relSt is a dominant strategy for \emph{all} publishers, it is still true that the \relSt is dominant for publisher 1:

\begin{lemma}\label{lem:xmmr first dominant}
    In any xMMR \modelName, the \relSt is a dominant strategy for publisher 1.
\end{lemma}

Note that Lemma~\ref{lem:xmmr first dominant} does not require the symmetry of the aspect distribution. A full proof of the lemma can be found in Appendix~\ref{proof lem:xmmr first dominant}.
Relying on Observation~\ref{obs:xmmr copycat zero} and Lemma~\ref{lem:xmmr first dominant}, one can derive a complete equilibrium characterization for the two-publisher case:

\begin{corollary}\label{cor:xmmr trivial pne}
    In any xMMR two-publisher game, $\strategyProfile$ is a Nash equilibrium if and only if $\strategyProfile_1$ is the \relSt and $\strategyProfile_2$ is a best reply of publisher 2, i.e.,
    $\strategyProfile_2 \in \argmaxOn{\strategyProfile_2}{\publisherUtilityFunction_2(\strategyProfile_1, \strategyProfile_2)}$.
    In addition, in any equilibrium $\strategyProfile_2 \neq \strategyProfile_1$.
\end{corollary}

Put differently, in the two-publisher case, xMMR ranking simultaneously achieves stability and diversity, as publishers reach a stable state in which their documents are not identical.
However, introducing a third publisher complicates the analysis, and equilibrium existence is no longer guaranteed even within the symmetric aspect distribution regime.

\begin{example}\label{exa:xmmr no pne}
    Let $\game$ be an xMMR three-publisher game with $\dimE = 1, \lambda = 0.5$ and the following symmetric aspect distribution:
    \begin{equation*}
        \aspectProb = 0.5 \cdot \ind(\aspectRep \in \cb{0, 0.1}).
    \end{equation*} 
    Then, $\game$ possesses no Nash equilibria.
\end{example}

The explanation is based on assuming the existence of an equilibrium by contradiction, and performing a detailed case analysis of the possible ranking outcome, showing that for all possible cases, a contradiction arises, with at least one publisher having a profitable deviation. Some of the analysis is done using numerical methods. Full explanation of the example appears in Appendix~\ref{app:example}.

To conclude, our analysis reveals that while in the two-publisher case, xMMR simultaneously satisfies stability and diversity, it generally induces environments in which diversity is often achieved at the expense of stability, and publishers are expected to frequently modify their documents rather than converge to a stable state when documents are non-identical.

\section{Utility-Induced Ranking}\label{sec:uir_ranking}

In \S\ref{sec:main results lex} we analyzed the \modelPlural induced by the two ranking functions, xQuAD and xMMR, in terms of stability and diversity. As highlighted in \S\ref{sec:model}, the publisher utilities in the games depend both on the ranking function (which determines the ranks given the strategy profile) and the user utility function $\clickModel$.

Our game-theoretic analysis reveals several drawbacks that arise while using the proposed ranking functions. Using the xQuAD ranker leads to a degenerate game in which all publishers choose the same document representation in the symmetric aspect distribution regime (Theorem~\ref{the:xquad lex dominant}). In the asymmetric case, an equilibrium might not exist, which can be interpreted as instability (Observation~\ref{obs:xquad non-sym no pne}). Under xMMR ranking, a diverse equilibrium may exist in the two-publisher case (Corollary~\ref{cor:xmmr trivial pne}), but even in the three-publisher case an equilibrium might not exist, even under aspect distribution symmetry (Example~\ref{exa:xmmr no pne}). Overall, we conclude that under both retrieval methods, simultaneously achieving stability and diversity is extremely difficult when publishers are strategic.

We now turn to present a novel general approach, termed \emph{utility induced ranking} (\textbf{UIR}) to resolving the issue just discussed. Specifically, we present a method of inducing a retrieval score function based on a given user utility function. Thus, each user utility function essentially entails a retrieval score function that guarantees the existence of an equilibrium in games based on this user utility function.
Note that using a UIR retrieval score function does not, in general, guarantee that the induced game admits a diverse equilibrium. In the remainder of this section, we begin by introducing our UIR framework and then we present our general result. Next, we demonstrate that when applying UIR to the aspect-coverage-based and novelty-based user utility functions, the resulting ranking functions can be viewed as variants of the original xQuAD and xMMR, respectively. For xQuAD, we demonstrate that our variant enables a diverse equilibrium under aspect-distribution asymmetry. However, the lack of diversity for symmetric distributions remains. For xMMR, we show that our variant guarantees stability via equilibrium existence. We now define the notion of utility-induced ranking (UIR).

\begin{definition}[Utility-induced ranking]
\label{def:uir_score_fn}
Let $\clickModel$ be a monotone user utility function, such that $\clickModel(\rankIndex, \simAllPublishersRank)$ does not depend on $\simAllPublishersRank[\rankIndex']$ for $\rankIndex' > \rankIndex$.
Then, based on $\clickModel$, the utility-induced ranking (UIR) is the ranking induced by using the following retrieval score function:
\begin{equation*}
        \scoreFunc_{\clickModel}(\strategy, \aspectDistribution, \partialRank) = \aspectEShort{\clickModel(\abs{\partialRank} + 1, \rb{\simPublisherTag}_{\strategyTag \in \partialRank} \mathbin{\|} \simPublisher)}.
\end{equation*}
\end{definition}

Importantly, both the aspect-coverage-based and novelty-based user utility functions satisfy the requirements of Definition~\ref{def:uir_score_fn}.
Our main result regarding utility-induced ranking can now be stated:

\begin{theorem}\label{the:direct pne}
    Let $\clickModel$ be a user utility function satisfying the conditions specified in Definition~\ref{def:uir_score_fn}, and let $\game$ be a \modelName induced by $\clickModel$ and its associated UIR function $\rankFunction_{\clickModel}$.
    Then, $\game$ has at least one equilibrium described by the following recursive formula:
    \begin{equation*}
        \strategy \in \argmaxOn{\strategy \in \embeddingSpaceBase}{\scoreFunc\rb{\strategy, \aspectDistribution, \set{\strategyTag \given \publisherIndex' < \publisherIndex}}}.
    \end{equation*}
    In addition, if the maximizer of each expression is unique, then the Nash equilibrium is unique.
\end{theorem}

The proof relies on the concept of iterated removal of (weakly) dominated strategies, a widely used concept in game theory.

\begin{proof}
    Let $\clickModel$ be a monotone user utility function, such that $\clickModel(\rankIndex, \simAllPublishersRank)$ does not depend on $\simAllPublishersRank[\rankIndex']$ for $\rankIndex' > \rankIndex$. In addition, let $\game$ be a \modelName induced by the user utility function $\clickModel$ and its associated UIR function $\rankFunction_{\clickModel}$.

    \paragraph{Part 1 - equilibrium formula}
    We will show that the strategy profile $\strategyProfile$ defined by recursively choosing \begin{equation*}
        \strategy \in \argmaxOn{\strategy \in \embeddingSpaceBase}{\scoreFunc\rb{\strategy, \aspectDistribution, \set{\strategyTag \given \publisherIndex' < \publisherIndex}}}
    \end{equation*} is a Nash equilibrium.

    We start by proving, by induction on $\publisherIndex$, that for every $\publisherIndex \in [\numPublishers]$, publisher $\publisherIndex$'s position in the ranking is her index, i.e., $\publisherRankLocation = \publisherIndex$. For $\publisherIndex = 1$, we know that $\strategyProfile_1 \in \argmaxOn{\strategyProfile_1 \in \embeddingSpaceBase}{\scoreFunc\rb{\strategyProfile_1, \aspectDistribution, \emptyset}}$. Hence, due to the lexicographic tie-breaker, $\publisherRankLocationSpecific{1} = 1$. Assume that $\publisherRankLocationTag = \publisherIndex'$ for every $\publisherIndex' < \publisherIndex$.
    Therefore, $\partialRank = \set{\strategyTag \given \publisherIndex' < \publisherIndex}$ and by definition $\strategy$ is a maximizer of the score function. Hence, due to the lexicographic tie-breaker, $\publisherRankLocation = \publisherIndex$.

    Let us fix some $\publisherIndex$ to be a publisher in the game. Notice that:
    \begin{equation}\label{equ:utility development 1}
    \begin{aligned}
        \publisherUtility(\strategy, \strategiesOthers) & =
        \aspectEShort{\clickModel(\publisherRankLocation, \simAllPublishersRank)}
        \\ &
        \overset{(1)}{=} \aspectEShort{\clickModel(\publisherIndex, \simAllPublishersRank)}
        \\ &
        \overset{(2)}{=} \aspectEShort{\clickModel(\publisherIndex, \rb{\simAllPublishersRank[\rankIndex]}_{\rankIndex \leq \publisherIndex})}
        \\ &
        \overset{(3)}{=} \aspectEShort{\clickModel(\publisherIndex, \rb{\simPublisherTag}_{\publisherIndex' < \publisherIndex} \mathbin{\|} \simFunction(\strategy, \aspectRep))}
        \\ &
        \overset{(4)}{=} \scoreFunc_{\clickModel}(\strategy, \aspectDistribution, \set{\strategyTag \given \publisherIndex' < \publisherIndex})
    \end{aligned}
    \end{equation}
    when:
    \begin{itemize}
        \item Transitions~(1) and (3) are based on the proof by induction for publishers 1 to $\publisherIndex$.
        \item Transition~(2) follows from the independence in lower positions.
        \item Transition~(4) follows from the utility-induced score function's definition.
    \end{itemize}
    
    Suppose by contradiction that she has a strategy $\strategyTilde$ such that $\publisherUtility(\strategy, \strategiesOthers) < \publisherUtility(\strategyTilde, \strategiesOthers)$. We will denote the new ranking in this position by $\rankFunctionTilde$, and the induced similarity scores for aspect $\aspect$ by $\simAllPublishersRankTilde$. In addition, we will denote by $\simFunctionGroup^*$ the induced similarity scores under the $(\strategyTilde, \strategiesOthers)$ profile after switching the positions of publisher $\publisherIndex$ and the publisher ranked in position $\publisherIndex$ (if they are the same publisher, then $\simFunctionGroup^* = \simAllPublishersRankTilde)$.
    
    Note that for every $\publisherIndex' < \publisherIndex$, $\publisherRankTildeLocationTag = \publisherIndex'$, as the proof we showed in the induction earlier is still valid for those positions. Therefore, $\publisherRankTildeLocation \geq \publisherIndex = \publisherRankLocation$. 
    Using this, we get that:
    \begin{equation}\label{equ:utility development 2}
    \begin{aligned}
        \publisherUtility(\strategyTag, \strategiesOthers) & =
        \aspectEShort{\clickModel(\publisherRankTildeLocationTag, \simAllPublishersRankTilde)}
        \\ &
        \overset{(1)}{\leq} \aspectEShort{\clickModel(\publisherIndex, \simFunctionGroup^*)}
        \\ &
        \overset{(2)}{=} \aspectEShort{\clickModel(\publisherIndex, \rb{\simFunctionGroup^*[\rankIndex]}_{\rankIndex \leq \publisherIndex})}
        \\ &
        \overset{(3)}{=} \aspectEShort{\clickModel(\publisherIndex, \rb{\simPublisherTag}_{\publisherIndex' < \publisherIndex} \mathbin{\|} \simFunction(\strategyTilde, \aspectRep))}
        \\ &
        \overset{(4)}{=} \scoreFunc_{\clickModel}(\strategyTilde, \aspectDistribution, \set{\strategyTag \given \publisherIndex' < \publisherIndex})
    \end{aligned}
    \end{equation}
    when:
    \begin{itemize}
        \item Transition~(1) follows from the monotony of $\clickModel$.
        \item Transition~(2) follows from the independence in lower positions.
        \item Transition~(3) follows from the fact that the top $\publisherIndex - 1$ positions are the first $\publisherIndex - 1$ publishers.
        \item Transition~(4) and follow from the utility-induced score function's definition.
    \end{itemize}

    Combining Equations~\eqref{equ:utility development 1} and~\eqref{equ:utility development 2} with the choice of
    \begin{equation*}
        \strategy \in \argmaxOn{\strategy \in \embeddingSpaceBase}{\scoreFunc\rb{\strategy, \aspectDistribution, \set{\strategyTag \given \publisherIndex' < \publisherIndex}}},
    \end{equation*}
    and we get that:
    \begin{equation*}
        \publisherUtility(\strategy, \strategiesOthers) \geq \publisherUtility(\strategyTag, \strategiesOthers).
    \end{equation*}
    This result contradicts the assumption that $\publisherUtility(\strategy, \strategiesOthers) < \publisherUtility(\strategyTilde, \strategiesOthers)$. Therefore, $\strategyProfile$ is a Nash equilibrium.

    One can interpret the strategy selection process as follows. The first publisher chooses a strategy that is optimal for them (a dominant strategy). Given this choice, the second publisher can restrict attention to the subset of strategies that maximize their utility given the first publisher's strategy. This reasoning proceeds analogously for all publishers. The resultant procedure corresponds to the iterated elimination of (weakly) dominated strategies, and the resulting strategy profile constitutes an equilibrium.

    \paragraph{Part 2 - equilibrium uniqueness}
    We now turn to prove the second part of the theorem. Assuming that the maximizer of each expression is unique. We will show that the Nash equilibrium is unique.
    
    Suppose for contradiction that there is more than one Nash equilibria. The maximizer of each expression $\scoreFunc\rb{\strategy, \aspectDistribution, \set{\strategyTag \given \publisherIndex' < \publisherIndex}}$ is unique. Hence, there is a Nash equilibrium denoted $\tilde{\strategyProfile}$ such that $\tilde{\strategyProfile} \neq \strategyProfile$, when $\strategyProfile$ is defined as before. Let $\publisherIndex$ be the first publisher lexicographically such that $\strategyTilde \neq \strategy$.

    Focusing on Equations~\eqref{equ:utility development 1} and~\eqref{equ:utility development 2}, we can see that all transitions are still valid as none of them uses any information about $\strategyTag$ for $\publisherIndex' > \publisherIndex$.
    
    In addition, using the new assumption, we now know that the $\strategy$ is the unique maximizer of $\scoreFunc\rb{\strategy, \aspectDistribution, \set{\strategyTag \given \publisherIndex' < \publisherIndex}}$. Therefore:
    \begin{gather*}
        \scoreFunc_{\clickModel}(\strategy, \aspectDistribution, \set{\strategyTag \given \publisherIndex' < \publisherIndex}) >
        \scoreFunc_{\clickModel}(\strategyTilde, \aspectDistribution, \set{\strategyTag \given \publisherIndex' < \publisherIndex})
        \implies
        \\
        \implies
        \publisherUtility(\strategy, \strategiesOthers) >
        \publisherUtility(\strategyTilde, \strategiesOthers),
    \end{gather*}
    which means that there is a profitable deviation, and $\tilde{\strategyProfile}$ is not a Nash equilibrium.
\end{proof}

We now turn to demonstrate an application of our general UIR framework to the aspect-coverage-based and novelty-based user utility functions. Substituting those specified in \S\ref{sub:user_utility} into Definition~\ref{def:uir_score_fn} yields the following variants of the xQuAD and xMMR score functions:

\begin{equation*}
    \scoreFunc_{UIR-xQuAD}(\strategy, \aspectDistribution, \partialRank) = \aspectEShort{\simPublisher \cdot \underset{\strategyTag \in \partialRank}{\Pi} \big(1 - \simPublisherTag\big)},
\end{equation*}

\begin{equation*}
    \scoreFunc_{UIR-xMMR}(\strategy, \aspectDistribution, \partialRank)
    = \aspectEShort{\simPublisher \cdot \minOn{\strategyTag \in \partialRank}{{(\simPublisher - \simPublisherTag)}^2}}.
\end{equation*}

Using these score functions, we define two corresponding game families: the UIR-xQuAD family, which employs the aspect-coverage-based user utility function, and the UIR-xMMR family, which employs the novelty-based user utility function.

Using Theorem~\ref{the:direct pne}, we can now analyze game instances with the two variants of xQuAD and xMMR ranking functions.
In particular, the following lemma demonstrates how for the UIR-xQuAD retrieval score function, a unique degenerate equilibrium arises under a symmetric aspect distribution, and a non-degenerate equilibrium arises under a non-symmetric aspect distribution.

\begin{lemma}\label{lem:xquad uir unique pne}
    Let $\game$ be a UIR-xQuAD game with $\aspectDistributionSuppSize \geq 2$. Then, $\game$ admits a unique Nash equilibrium. 
    
    Furthermore, the following properties hold:
    \begin{itemize}
        \item If $\aspectDistribution$ is a symmetric distribution, then the \relSt is a dominant strategy for each publisher; thus, the \relPro is the unique equilibrium.
        \item If $\aspectDistribution$ is not a symmetric distribution and $\aspectDistributionSuppSize = 2$, then under the equilibrium, the documents are distinct, i.e., $\strategyProfile_1 \neq \strategyProfile_2$.
    \end{itemize}
\end{lemma}

The proof of the equilibrium uniqueness is very technical, as we develop a closed formula for the equilibrium strategies based on generalize version of Proposition~\ref{pro:sim maximizer}. The proof of the symmetric case is similar to the proof presented in Theorem~\ref{the:xquad lex dominant}, while the proof of the non-symmetric case is based on finding a closed formula for the equilibrium strategies. A full proof of the lemma can be found in Appendix~\ref{proof lem:xmmr uir pne}. 

The advantage of the existence of a unique equilibrium is that one can predict in advance which equilibrium the system will converge to (in contrast to the case of multiple equilibria, where a certain degree of uncertainty arises).

Moreover, although we lack a theoretical guarantee regarding the diversity of the equilibrium in the case of non-symmetric aspect distributions with $\aspectDistributionSuppSize > 2$, we explored this property empirically. We sampled various aspect distributions for different values of $\numPublishers, \dimE$, and $\aspectDistributionSuppSize$, and calculated the unique Nash equilibrium of the induced UIR-xQuAD game. Our findings show that in all cases, there was at least one pair of distinct documents, and in $99\%$ of the cases, every document in the equilibrium was unique\footnote{We used $\numPublishers \in \{2, \ldots, 10\}, \dimE \in \{1, \ldots, 11\}$, and $\aspectDistributionSuppSize \in \{2, \ldots, 10\}$. For each parameter combination, we sampled $1,000$ aspect distributions with random supports and random probabilities.}.

In addition, when focusing on the UIR-xMMR, we attain a guarantee of stability and diversity, described formally by the following lemma:
\begin{lemma}\label{lem:xmmr uir pne}
    Let $\game$ be a UIR-xMMR \modelName. Then, $\game$ admits at least one Nash equilibrium. In addition, for every equilibrium profile $\strategyProfile$ and for every two publishers $\publisherIndex \neq \publisherIndex'$, $\strategyProfile_\publisherIndex \neq \strategyProfile_{\publisherIndex'}$.
\end{lemma}

In addition to using Theorem~\ref{the:direct pne}, the proof is also based on technical developments that were presented as part of the proof of Observation~\ref{obs:xmmr copycat zero}. A full proof of the lemma can be found in Appendix~\ref{proof lem:xmmr uir pne}.

\paragraph{Convergence of learning dynamics}
Beyond the guaranteed existence of equilibrium, we examined the convergence of learning dynamics for the two variations we presented, UIR-xQuAD and UIR-xMMR. The convergence of learning dynamics is a critical concern, as the mere existence of a theoretical equilibrium does not guarantee that players (publishers in our context) will naturally adopt it. Demonstrating empirical convergence validates that when publishers engage in the defined game sequentially through a learning process, they trend toward an equilibrium. This empirical evidence bridges the gap between abstract game theoretic stability and the actual, predictable behavior of participants in a dynamic environment.

For this purpose, we used the discrete better response dynamics algorithm presented in~\citet{madmon2023search}. In this algorithm, we initialize from a specific profile. In each round, a random publisher is selected to modify their strategy to increase their utility, where the available strategic adjustments are based on a fixed set of step sizes and movement directions in the document's embedding space (algorithm parameters). We define this dynamics as converged when no publisher can deviate to improve their utility by $\epsilon$ (an algorithm parameter) or more.

Using the algorithm, we found that for both games we obtain $100\%$ convergence of the learning dynamics for different values of $\numPublishers, \dimE$ and $\aspectDistributionSupp$, while sampling different aspects' distributions and starting profiles. This means that not only do these games have an equilibrium, but the players can also reach it by a series of logical actions\footnote{We used $\numPublishers \in \{2, \ldots, 10\}, \dimE \in \{1, \ldots, 5\}$, and $\aspectDistributionSuppSize \in \{2, \ldots, 10\}$. For each parameter combination, we sampled $100$ pairs, each pair consists of an aspect distribution with random supports and random probabilities, and a random starting profile (used in the algorithm). In addition, we used $\mathcal{S} = \{2^{-6}, 2^{-5}, \ldots, 2^{-2}, 0.5, 0.6, \ldots, 1 \}$, $\;\mathcal{D} = \big \{ \frac{d}{||d||_2}: d \in \{-1,0,1\}^k \setminus \{ \vec{0} \} \big \}$, $T=1000, \; M=900, \; \varepsilon=10^{-6}$, which are the same parameters used in~\cite{madmon2023search}.}.

\section{Conclusions and Future Work}\label{sec:conclusions}
We presented a game theoretic analysis of a competitive search setting where search-results diversification is applied. We showed that representative diversity-based ranking functions face a fundamental tradeoff under strategic publisher behavior. While the xQuAD method induces strong stability guarantees (for symmetric distributions) at the cost of degenerate, anti-diverse outcomes, the xMMR method may incentivize diverse content but fail to admit equilibria (stability). Our novel utility-induced ranking framework yields diversity-based ranking functions that can lead to both stability and diversity.

Our analysis also gives rise to important directions for future work. We focused on a binary measure of stability, specifically, whether all documents in the equilibrium profile are identical. Future work could define numerical measures of diversity, such as the variance of document representations, to evaluate various aspects of diversity across different ranking functions.

Our main results focus on deterministic, score-based ranking functions with lexicographical tie-breaking. Our analysis can be extended to address \emph{randomized ranking functions}, which will allow, for example, a random tie breaker rule. In addition, it is an open question whether randomized ranking can fundamentally change publishers' incentives and mitigate the diversity-stability tradeoff identified in our work. The exploration of such mechanisms remains a natural and important direction for future work.

\begin{acks}
We thank the reviewers for their comments. The paper is based on work supported in part by the Israel Science Foundation (grant no. 403/22).
\end{acks}

\bibliographystyle{ACM-Reference-Format}
\balance 
\bibliography{references}

\appendix

\section{Omitted Proofs}\label{app:proofs}

Before presenting the proofs of the claims in the main text, we introduce the following auxiliary lemma, which will be used several times in the proofs that follow.

\begin{lemma}\label{lem:aux two aspects same weight}
    Let $\aspectDistribution$ be an aspect distribution with $\aspectDistributionSupp = \{\aspect_1, \aspect_2\}$, such that $\aspect_1 \neq \aspect_2$. The following three conditions are equivalent:
    \begin{enumerate}
        \item $\aspectDistribution$ is symmetric.
        \item $\norm{\aspectEShort{\aspectRep} - \aspectRepOne} = \norm{\aspectEShort{\aspectRep} - \aspectRepTwo}$.
        \item $\aspectProbSpecific{1} = \aspectProbSpecific{2} = 0.5$
        \item $\norm{\aspectEShort{\aspectRep} - \aspectRepOne} = \frac{1}{2}\norm{\aspectRepOne - \aspectRepTwo}$.
    \end{enumerate}
\end{lemma}

\begin{proof}
    We prove the equivalence by showing $(1) \implies (2)$, $(2) \implies (1)$, and then showing both $(2) \iff (3)$ and $(3) \iff (4)$ together.
    
    \textbf{(1) $\implies$ (2):}  
    Assume $\aspectDistribution$ is symmetric.
    Suppose for contradiction that $\norm{\aspectEShort{\aspectRep} - \aspectRepOne} \neq \norm{\aspectEShort{\aspectRep} - \aspectRepTwo}$. By the symmetric distribution definition:
    \begin{align*}
        \aspectRepOne &
        = \aspectEShort{\aspectRep \mid \norm{\aspectEShort{\aspectRep} - \aspectRep} = \norm{\aspectEShort{\aspectRep} - \aspectRepOne} }
        = \\ & =
        \aspectEShort{\aspectRep \mid \norm{\aspectEShort{\aspectRep} - \aspectRep} = \norm{\aspectEShort{\aspectRep} - \aspectRepTwo}} = \aspectRepTwo
    ,\end{align*}
    and we reached a contradiction. Therefore, $\norm{\aspectEShort{\aspectRep} - \aspectRepOne} = \norm{\aspectEShort{\aspectRep} - \aspectRepTwo}$.

    \textbf{(2) $\implies$ (1):}  
    If $\norm{\aspectEShort{\aspectRep} - \aspectRepOne} = \norm{\aspectEShort{\aspectRep} - \aspectRepTwo}$, then the distance $\norm{\aspectEShort{\aspectRep} - \aspectRep}$ is equal for both aspects, therefore
    \begin{equation*}
        \aspectEShort{\aspectRep \mid \norm{\aspectEShort{\aspectRep} - \aspectRep} = \aspectD} = \aspectEShort{\aspectRep}.
    \end{equation*}

    \textbf{(2) $\iff$ (3) and (3) $\iff$ (4):}  
    Using the expected value definition:

    \begin{gather*}
        \begin{aligned}
        \norm{\aspectEShort{\aspectRep} - \aspectRepOne} & = \norm{\aspectRepOne \cdot \aspectProbSpecific{1} + \aspectRepTwo \cdot \aspectProbSpecific{2} - \aspectRepOne}
        \\ & = \norm{\aspectRepTwo \cdot \aspectProbSpecific{2} - \aspectRepOne \cdot (1 - \aspectProbSpecific{1})} 
        \\ & = 
        \norm{\aspectRepTwo \cdot \aspectProbSpecific{2} - \aspectRepOne \cdot \aspectProbSpecific{2}}
        \\ & = 
        \aspectProbSpecific{2} \norm{\aspectRepTwo - \aspectRepOne}
        \end{aligned}
        \\
        \begin{aligned}
        \norm{\aspectEShort{\aspectRep} - \aspectRepTwo} & = \norm{\aspectRepOne \cdot \aspectProbSpecific{1} + \aspectRepTwo \cdot \aspectProbSpecific{2} - \aspectRepTwo}
        \\ & = \norm{\aspectRepOne \cdot \aspectProbSpecific{1} - \aspectRepTwo \cdot (1 - \aspectProbSpecific{2})} 
        \\ & = 
        \norm{\aspectRepOne \cdot \aspectProbSpecific{1} - \aspectRepTwo \cdot \aspectProbSpecific{1}}
        \\ & = 
        \aspectProbSpecific{1} \norm{\aspectRepOne - \aspectRepTwo}
        = \aspectProbSpecific{1} \norm{\aspectRepTwo - \aspectRepOne}
        \end{aligned}
    \end{gather*}
    
    Therefore, $\norm{\aspectEShort{\aspectRep} - \aspectRepOne} = \norm{\aspectEShort{\aspectRep} - \aspectRepTwo}$ if and only if $\aspectProbSpecific{1} = \aspectProbSpecific{2} = 0.5$.

    In addition, if $\norm{\aspectEShort{\aspectRep} - \aspectRepOne} = \aspectProbSpecific{2} \norm{\aspectRepTwo - \aspectRepOne}$, then $\aspectProbSpecific{1} = \aspectProbSpecific{2} = 0.5$ is true if and only if $\norm{\aspectEShort{\aspectRep} - \aspectRepOne} = \frac{1}{2}\norm{\aspectRepOne - \aspectRepTwo}$ is true.
    
    Combining the three parts, we conclude that the four conditions are equivalent.
\end{proof}

\subsection{Lemma~\ref{lem:xquad recursive maximizer}}\label{proof lem:xquad recursive maximizer}
\begin{proof}    
    Let $\aspectDistribution$ be a symmetric aspect distribution, and let $\publishersSet' \subseteq \publishersSet$ be a set of publishers. Assuming that $\forall \publisherIndex' \in \publishersSet', \; \strategyTag = \aspectEShort{\aspectRep}$, we will show that:
    \begin{equation*}
        \argmaxOn{\strategy}{\aspectEShort{\simPublisher \cdot \underset{\publisherIndex' \in \publishersSet'}{\Pi} \big(1 - \simPublisherTag\big)}} = \aspectEShort{\aspectRep}.
    \end{equation*}
    We denote $\simFunction(\aspectEShort{\aspectRep}, \aspectRep)=\simExpected$, and we get that:
    \begin{gather*}
        \argmaxOn{\strategy}{\aspectEShort{\simPublisher \cdot \underset{\publisherIndex' \in \publishersSet'}{\Pi} \big(1 - \simPublisherTag\big)}}
        = \\ =
        \argmaxOn{\strategy}{\aspectEShort{\simPublisher \cdot \underset{\publisherIndex' \in \publishersSet'}{\Pi} \big(1 - \simExpected\big)}}
        = \\ = 
        \argmaxOn{\strategy}{\aspectEShort{\simPublisher \cdot (1 - \simExpected)^{\abs{\publishersSet'}}}}
    .\end{gather*}
    
    Let $\aspectD \geq 0$ such that there is some aspect representation $\aspectRep$ such that $\norm{\aspectEShort{\aspectRep} - \aspectRep} = \aspectD$. Notice that for every aspect:
    \begin{equation*}
        \norm{\aspectEShort{\aspectRep} - \aspectRep} = \aspectD \implies (1 - \simExpected) = \frac{1}{\dimE}\norm{\aspectEShort{\aspectRep} - \aspectRep}^2 = \frac{\aspectD^2}{\dimE}.
    \end{equation*}
    Using the symmetric aspect distribution definition (Definition~\ref{def:sym dist}), we get that:
    \begin{gather*}
        \argmaxOn{\strategy}{\aspectEShort{\simPublisher \cdot (1 - \simExpected)^{\abs{\publishersSet'}} \mid \norm{\aspectEShort{\aspectRep} - \aspectRep} = \aspectD}}
        \\
        \overset{(1)}{=} \argmaxOn{\strategy}{\aspectEShort{\simPublisher \cdot {(\frac{\aspectD^2}{\dimE})}^{\abs{\publishersSet'}} \mid \norm{\aspectEShort{\aspectRep} - \aspectRep} = \aspectD}}
        \\
        = \argmaxOn{\strategy}{\aspectEShort{\simPublisher \mid \norm{\aspectEShort{\aspectRep} - \aspectRep} = \aspectD}} 
        \\
        \overset{(2)}{=} \aspectEShort{\aspectRep \mid \norm{\aspectEShort{\aspectRep} - \aspectRep} = \aspectD }
        \overset{(3)}{=} \aspectEShort{\aspectRep}
    ,\end{gather*}
    
    where transition~(1) follows from the fact that ${\left(\frac{\aspectD^2}{\dimE}\right)}^{\abs{\publishersSet'}}$ does not depend on $\strategy$,
    transition~(2) follows from Proposition~\ref{pro:sim maximizer},
    and transition~(3) follows from the definition of a symmetric aspect distribution.
    Since $\aspectEShort{\aspectRep}$ is the unique maximizer for each $\aspectD$ independently, it is also the unique maximizer for the unconditional distribution $\aspectDistribution$.
\end{proof}

\subsection{Lemma~\ref{lem:xquad lex relevance better rank}}\label{proof lem:xquad lex relevance better rank}
We will present a more general lemma and prove it, which would also provide proof for the original lemma.

\begin{lemma}\label{lem:xquad general relevance better rank}
     Let $\game$ be a \modelName with a symmetric aspect distribution. 
     Assume that $\rankFunction$ is a score-based ranking function that uses the xQuAD score with $\lambda \in \bb{0,1}$ and an arbitrary tie-breaking rule (not necessarily lexicographic).
     For every $\publisherIndex, \publisherIndex' \in \publishersSet$ s.t $\strategy = \aspectEShort{\aspectRep}$ and $\strategyTag \neq \aspectEShort{\aspectRep}$, it holds that $\publisherRankLocation < \publisherRankLocationTag$.
\end{lemma}

Notice that in the standard version of the xQuAD score function $\lambda \in \rb{0, 1}$, and we would prove the lemma for a bigger range.

\begin{proof}
    Remember that the score function is:
    \begin{align*}
        \scoreFunc_{xQuAD}(\strategy, \aspectDistribution, \partialRank) & =    
        \aspectEShort{\lambda\simPublisher + (1 - \lambda)\simPublisher \cdot \underset{\strategyTag \in \partialRank}{\Pi} \big(1 - \simPublisherTag\big)}
        \\ & = \lambda\aspectEShort{\simPublisher} + (1 - \lambda) \aspectEShort{\simPublisher \cdot \underset{\strategyTag \in \partialRank}{\Pi} \big(1 - \simPublisherTag\big)}
    .\end{align*}

    By Proposition~\ref{pro:sim maximizer}, $\aspectEShort{\aspectRep}$ is the unique maximizer of the first component, and by Lemma~\ref{lem:xquad recursive maximizer}, $\aspectEShort{\aspectRep}$ is the unique maximizer of the first component. Therefore, $\aspectEShort{\aspectRep}$ is the unique maximizer of the score function.
    
    Suppose for contradiction that the claim is not true and that there is one publisher or more who do not play $\aspectEShort{\aspectRep}$ and have a better rank than $\publisherIndex$. Let $\publisherIndex'$ be the highest-ranked publisher that does not play $\aspectEShort{\aspectRep}$. We will denote the publishers that were ranked before publisher $\publisherIndex'$ by $\partialRank^{\publisherRankLocationTag}$ (it can be the null group if $\publisherRankLocationTag$ is 1). Notice that by the definition of $\publisherIndex'$, $\forall \publisherIndex'' \in \partialRank^{\publisherRankLocationTag}, \; \strategyTagTag = \aspectEShort{\aspectRep}$. Therefore, we reached a contradiction as publisher $\publisherIndex$ has a better score than publisher $\publisherIndex'$, which means that $\publisherIndex'$ should not have been ranked at this position.
\end{proof}

We proved the general case lemma, then the original Lemma~\ref{lem:xquad lex relevance better rank} in which the tie-breaker is lexicographic is also true.

\subsection{Observation~\ref{obs:xquad non-sym no pne}}\label{proof obs:xquad non-sym no pne}
During the proof, we will use the following definition
\begin{definition}\label{def:best response}
    A strategy $\strategy \in \embeddingSpace$ is a best response to $\strategiesOthers \in \prod_{\publisherIndex' \neq \publisherIndex} \embeddingSpaceSpecific{\publisherIndex'}$ if $\; \forall \strategy' \in \embeddingSpace, \; \publisherUtility(\strategy, \strategiesOthers) \geq \publisherUtility(\strategy', \strategiesOthers)$.
\end{definition}
\begin{proof}
    Let $\game$ be an xQuAD three-publisher game with $\dimE = 1, \lambda = 0.5$ and the following non-symmetric aspect distribution:
    \begin{equation*}
        \aspectProb = 0.75 \cdot \ind(\aspectRep = 0) + 0.25 \cdot \ind(\aspectRep = 1).
    \end{equation*}

     Before we begin, note that the distribution we have presented fails to satisfy the condition equivalent to symmetry (Auxiliary Lemma~\ref{lem:aux two aspects same weight}) and is therefore non-symmetric.

    Suppose for contradiction that there is some profile of $\strategyProfile = (\strategyProfile_1, \strategyProfile_2, \strategyProfile_3)$ which is NE.

    We start by showing that playing the \relSt of $\aspectEShort{\aspectRep}$ is a dominant strategy for publisher 1. 
    First, notice that in the xQuAD game, the user utility function is monotone, as the elements in $\underset{\rankIndex' < \rankIndex}{\Pi} \big(1 - \simLocationTag\big)$ are smaller than or equal to $1$. Second, as the \relSt is the unique maximizer of $\aspectEShort{\simPublisher}$, we can deduce that if publisher 1 plays it, then with the lexicographic tie breaker, she can secure the first position (no matter what the strategies of the other publishers). Combining this with the fact that the expected gain of the user from the first position is also $\aspectEShort{\simPublisher}$ and the monotony of the user utility function, we can conclude that the \relSt is a dominant strategy for publisher 1.

    If \relSt is a dominant strategy for publisher 1, then in NE, $\strategyProfile_1 = \aspectEShort{\aspectRep} = 0.25$. Based on that, we can find the score function for the second-place comparison and the utility function of the second-ranked publisher:
    \begin{align*}
        \scoreFunc_{xQuAD} & (\strategy; \aspectDistribution, \strategyProfile_1) =
        \frac{\simFunction^{\aspect_1}_\publisherIndex +  {\simFunction^{\aspect_1}_\publisherIndex \cdot (1 - \simFunction^{\aspect_1}_1)}}{2 \cdot \frac{1}{0.75}} + 
        \frac{\simFunction^{\aspect_2}_{\rankLocationSpecific{2}} +  {\simFunction^{\aspect_2}_{\rankLocationSpecific{2}} \cdot (1 - \simFunction^{\aspect_2}_1)}}{2 \cdot \frac{1}{0.25}}
        \\ & =
        \tfrac{3}{8} \Bigr( \bigl(1 - (\strategy)^2\bigr) \cdot \tfrac{17}{16} \Bigr) + 
        \tfrac{1}{8} \Bigr( \bigl(1 - ((\strategy - 1)^2\bigr) \cdot \tfrac{25}{16} \Bigr)
        \\ & =
        - \tfrac{19}{32}(\strategy)^2
        + \tfrac{25}{64}\strategy
        + \tfrac{51}{128}
    ,\end{align*}
    \begin{align*}
        \publisherUtilityFunction_{\rankLocationSpecific{2}} & (\strategyProfile_{\rankLocationSpecific{2}}; \strategyProfile_1) = 
        0.75 \cdot \Bigl(\simFunction^{\aspect_1}_{\rankLocationSpecific{2}} \cdot (1 - \simFunction^{\aspect_1}1)\Bigr) + 
        0.25 \cdot \Bigl(\simFunction^{\aspect_2}_{\rankLocationSpecific{2}} \cdot (1 - \simFunction^{\aspect_2}_1)\Bigr)
        \\ & =
        0.75 \cdot \Bigl((1 - (\strategyProfile_{\rankLocationSpecific{2}})^2) \cdot \tfrac{1}{16}\Bigr) + 
        0.25 \cdot \Bigl((1 - (\strategyProfile_{\rankLocationSpecific{2}} - 1)^2) \cdot \tfrac{9}{16}\Bigr)
        \\ & =
        -\tfrac{3}{16}(\strategyProfile_{\rankLocationSpecific{2}})^2
        + \tfrac{9}{32}\strategyProfile_{\rankLocationSpecific{2}}
        + \tfrac{3}{64}
    .\end{align*}

    Notice that $\scoreFunc_{xQuAD}(\strategy; \aspectDistribution, \strategyProfile_1)$ is a convex function with a global maximum at $-\tfrac{25 / 64}{-2 \cdot 19 / 32} = \tfrac{25}{76}$ and no additional local maximum. In addition, $\publisherUtilityFunction_{\rankLocationSpecific{2}}(\strategyProfile_{\rankLocationSpecific{2}}; \strategyProfile_1)$ is a convex function with a global maximum at $-\tfrac{9 / 32}{-2 \cdot 3 / 16} = \tfrac{3}{4}$ and no additional local maximum.

    Now, we will split into cases based on the strategies of publishers 2 and 3.
    \begin{itemize}
        \item If $\publisherUtilityFunction_2(\strategyProfile) < \publisherUtilityFunction_3(\strategyProfile)$: Publisher 2 can improve her utility by playing $\strategyProfile_3$ as in this case she will be ranked second (due to the tie-breaker). If publisher 3 is currently ranked second, then by playing $\strategyProfile_3$ publisher 2 will get utility of $\publisherUtilityFunction_3(\strategyProfile)$, and if publisher 3 is currently ranked third, publisher 2 will get utility bigger than or equal to $\publisherUtilityFunction_3(\strategyProfile)$, as the user utility function is monotone.
        \item If $\strategyProfile_2$ is the global maximizer of $\scoreFunc_{xQuAD}(\strategy; \aspectDistribution, \strategyProfile_1)$: 
        In this case, publisher 3 will be ranked third regardless of the strategy chosen. We get that:
        \begin{align*}
            \publisherUtilityFunction_{3} (\strategyProfile_3; \strategyProfile_1, \strategyProfile_2) & = 
            0.75 \cdot \Bigl(\simFunction^{\aspect_1}_3 \cdot (1 - \simFunction^{\aspect_1}_1) \cdot (1 - \simFunction^{\aspect_1}_2)\Bigr) 
            \\ & \quad \quad + 
            0.25 \cdot \Bigl(\simFunction^{\aspect_2}_1 \cdot (1 - \simFunction^{\aspect_2}_1) \cdot (1 - \simFunction^{\aspect_1}_2)\Bigr)
            \\ & =
            0.75 \cdot \Bigl((1 - (\strategyProfile_3)^2) \cdot \tfrac{1}{16} \cdot \tfrac{625}{5776} \Bigr)
            \\ & \quad \quad + 
            0.25 \cdot \Bigl((1 - (\strategyProfile_3 - 1)^2) \cdot \tfrac{9}{16} \cdot \tfrac{2601}{5776}\Bigr)
            \\ & = 
            \tfrac{1875}{369664} \cdot (1 - (\strategyProfile_3)^2)
            +
            \tfrac{23409}{369664} \cdot \Bigl(2 \strategyProfile_3 -(\strategyProfile_3)^2\Bigr)
            \\ & = 
            -\frac{25284}{369664}(\strategyProfile_3)^2
            + \frac{46818}{369664}\strategyProfile_3
            + \frac{1875}{369664}
        \end{align*}
        Therefore, publisher 3's best response is to play $\strategyProfile_3 = 1$. If publisher 3 does play the best response, then the score of publisher 2 is bigger than the score of publisher 3. This means that there is $0 < \epsilon$ such that if publisher 2 deviates and plays $\strategyProfile_2 + \epsilon$ her score is still bigger than publisher 3's score while her utility is increased, so this is a profitable deviation.
        \item If $\publisherUtilityFunction_2(\strategyProfile) > \publisherUtilityFunction_3(\strategyProfile)$ and $\strategyProfile_2$ is \textbf{not} the global maximizer of $\scoreFunc_{xQuAD}(\strategy; \aspectDistribution, \strategyProfile_1)$: As $\scoreFunc_{xQuAD}(\strategy; \aspectDistribution, \strategyProfile_1)$ has no local maximizer other the global maximizer, we know that for every $0 < \epsilon$ there is a strategy in $[\strategyProfile_2 - \epsilon, \strategyProfile_2 + \epsilon] \cap [0, 1]$ which results in a higher score value than the current score of publisher 2 (for the second position). Therefore, if we take $\epsilon \to 0$, for sufficiently small $\epsilon$, publisher 3 can deviate and get a higher score than publisher 2 while having a utility that converges to publisher 2's current utility (as the utility function is continuous if the rank is not changing), which must be a profitable deviation.
        \item If $\publisherUtilityFunction_2(\strategyProfile) = \publisherUtilityFunction_3(\strategyProfile)$ and publisher 2 is ranked second: Notice that in this profile publisher 3's utility is:
        \[
            \aspectEShort{\simPublisherSpecific{3} \cdot (1 - \simPublisherSpecific{1})(1 - \simPublisherSpecific{2})}.
        \] 
        In addition, it is easy to see that in $[0,1]$ there is no publisher 2's document for which both $\simFunction(\strategyProfile_2,0) = 1 - \norm{\strategyProfile_2}^2$ and $\simFunction(\strategyProfile_2,1) = 1 - \norm{1 - \strategyProfile_2}^2$ are $0$, i.e. one of them is positive. Therefore, 
        
        \begin{align*}
            \publisherUtilityFunction_2(\strategyProfile_1, \strategyProfile_2, \strategyProfile_3) &= 
            \publisherUtilityFunction_3(\strategyProfile_1, \strategyProfile_2, \strategyProfile_3) =
            \aspectEShort{\simPublisherSpecific{3} \cdot (1 - \simPublisherSpecific{1})(1 - \simPublisherSpecific{2})} \\&< 
            \aspectEShort{\simPublisherSpecific{3} \cdot (1 - \simPublisherSpecific{1})} = \publisherUtilityFunction_2(\strategyProfile_1, \strategyProfile_3, \strategyProfile_3),
        \end{align*}
        
        which means that a deviation to $\strategyProfile_3$ is a profitable deviation for publisher 2.
        \item If $\publisherUtilityFunction_2(\strategyProfile) = \publisherUtilityFunction_3(\strategyProfile)$, publisher 3 is ranked second and $\strategyProfile_3$ is the global maximizer of $\scoreFunc_{xQuAD}(\strategy; \aspectDistribution, \strategyProfile_1)$:
        If both publishers have the same utility values and publisher 3 plays the unique global-maximum strategy, then publisher 2 must also play the unique global-maximum strategy.
        This contradicts the fact that publisher 3 is ranked second, as in this strategy profile, publisher 2 will be ranked second due to the tie-breaker.
        \item If $\publisherUtilityFunction_2(\strategyProfile) = \publisherUtilityFunction_3(\strategyProfile)$, publisher 3 is ranked second and $\strategyProfile_3$ is \textbf{not} the global maximizer of $\scoreFunc_{xQuAD}(\strategy; \aspectDistribution, \strategyProfile_1)$:
        If publisher 3 is ranked second, then her score function value for the second position must be higher than publisher 3's score function value for the second position. Therefore, as $\strategyProfile_3$ is not a local maximum of the second-place utility function, and $\scoreFunc_{xQuAD}(\strategy; \aspectDistribution, \strategyProfile_1)$ is continuous, publisher 3 can move a little step of size $\epsilon$ to increase her utility, while keeping her score greater than publisher 2 score, and this is a profitable deviation.
    \end{itemize}

    We have shown that in every case there is a profitable deviation to at least one of the publishers and therefore we reached a contradiction and there is no Nash Equilibrium in game $\game$.
\end{proof}

\subsection{Observation~\ref{obs:xmmr copycat zero}}\label{proof obs:xmmr copycat zero}
In order to prove the observation, we will prove a more general lemma:
\begin{lemma}\label{lem:aux xmmr non zero strategy}
    Let $\game$ be \modelName with the novelty-based utility function 
    
    \[
        \clickModel^{novelty}(\rankIndex, \simAllPublishersRank) = \simLocation \cdot \minOn{\rankIndex' < \rankIndex}{{(\simLocation - \simLocationTag)}^2},
    \]
    
    and a ranking function $\rankFunction$ that is based on some score function $\scoreFunc$ with a lexicographic tie-breaker.
    In addition, let $\publisherIndex$ be some publisher and let $\strategiesOthers$ be some strategies for the other publishers.
    If there is a publisher $\publisherIndex' < \publisherIndex$ such that $\strategy = \strategyTag$ then $\publisherUtility(\strategy, \strategiesOthers) = 0$. In addition, there is some strategy $\strategyTilde$ with positive utility, i.e., $\publisherUtility^{xMMR}(\strategyTilde, \strategiesOthers) > 0$.
\end{lemma}

\begin{proof}
    Let $\game$ be \modelName with the novelty-based utility function, and a ranking function $\rankFunction$ that is based on some score function $\scoreFunc$ with a lexicographic tie-breaker.
    In addition, let $\publisherIndex$ be some publisher and let $\strategiesOthers$ be some strategies for the other publishers.
    
    \paragraph{Part 1 - utility calculation}
    Let $\publisherIndex'$ be a publisher such that $\publisherIndex' < \publisherIndex$ and $\strategy = \strategyTag$.
    Notice that for every score function, due to the lexicographic tie-breaker $\publisherRankLocation > \publisherRankLocationTag$. Therefore:
    \begin{align*}
        \publisherUtility^{xMMR}(\strategy, \strategiesOthers) 
        & =       
        \aspectEShort{\simPublisher \cdot \minOn{\rankIndex' < \publisherRankLocation}{(\simPublisher - \simLocationTag)^2}}
        \\
        & \leq \aspectEShort{\simPublisher \cdot (\simPublisher - \simPublisherTag)^2} = \aspectEShort{\simPublisher \cdot 0} = 0.
    \end{align*}
    
    \paragraph{Part 2 - equilibrium condition}
    Fix some aspect $\aspect_1 \in \aspectDistributionSupp$. In our model, $\simFunction(\strategy, \aspectRep) = 1 - \frac{1}{\dimE}\norm{\strategy-\aspectRep}^2$ therefore we know that after we subtract from $\embeddingSpace$ the strategies for which $\simFunction_{\publisherIndex}^{\aspect_1} \in \set{\simFunction_{\publisherIndex'}^{\aspect_1} \given \publisherIndex' \neq \publisherIndex}$, we still have an infinite number of possible strategies with $\simFunction_{\publisherIndex}^{\aspect_1} > 0$. We fix one of them and denote it by $\strategy$.
    
    \begin{align*}
        \publisherUtility^{xMMR}(\strategy, \strategiesOthers)
        & =       
        \aspectEShort{\simPublisher \cdot \minOn{\rankIndex' < \publisherRankLocation}{(\simPublisher - \simLocationTag)^2}}
        \\
        & \geq \aspectProbSpecific{1} \cdot \simPublisherSpecific{1} \cdot \minOn{\rankIndex' < \publisherRankLocation}{(\simPublisherSpecific{1} - \simLocationTag)^2}
        \\
        & > 0,
    \end{align*}
    where the last transition is based on that $\aspectProbSpecific{1} > 0$, and we chose $\strategy$ to be a strategy for which $\simPublisherSpecific{1}$ is positive and $\simPublisherSpecific{1} \notin \set{\simFunction_{\publisherIndex'}^{\aspect_1} \given \publisherIndex' \neq \publisherIndex}$ (if all terms in the min are positive, the min is also positive).
    
    Therefore, publisher $\publisherIndex$ has a strategy $\strategy$ with positive utility.

\end{proof}

We now proceed with the proof of the observation.

\begin{proof}
    Let $\game$ be an xMMR \modelName, and let $\publisherIndex > \publisherIndex'$ be two publishers. We fix some $\strategiesOthers$ to be the strategies for all publishers except publisher $\publisherIndex$.

    \paragraph{Claim 1}
    By Auxiliary Lemma~\ref{lem:aux xmmr non zero strategy}, if $\strategy = \strategyTag$ then $\publisherUtility^{xMMR}(\strategy, \strategiesOthers) = 0$.

    \paragraph{Claim 2}
    Let $\strategyProfile$ be some NE. Suppose by contradiction that $\strategy = \strategyTag$. Then, as we showed, $\publisherUtility^{xMMR}(\strategy, \strategiesOthers) = 0$. In addition, by Auxiliary Lemma~\ref{lem:aux xmmr non zero strategy} there is some strategy $\strategyTilde$ such that $\publisherUtility^{xMMR}(\strategyTilde, \strategiesOthers) > 0$. Therefore, publisher $\publisherIndex$ has a profitable deviation, $\strategyProfile$ is not a Nash equilibrium, and we reached a contradiction.
\end{proof}

\subsection{Lemma~\ref{lem:xmmr first dominant}}\label{proof lem:xmmr first dominant}
\begin{proof}
    Let $\game$ be an xMMR game, let $\strategyProfile_{-1}$ be strategies of all publishers except publisher 1, and let $\strategyProfile_1,\strategyProfile_1' \in \embeddingSpace$ such that $\strategyProfile_1 = \aspectEShort{\aspectRep}$ and $\strategyProfile_1' \neq \aspectEShort{\aspectRep}$.

    In Proposition~\ref{pro:sim maximizer} we showed that $\aspectEShort{\simPublisher}$ is uniquely maximized by $\aspectEShort{\aspectRep}$, so by playing $\aspectEShort{\aspectRep}$, publisher 1 can ensure the first position, as she wins in the lexicographic tie breaker in case of a tie. In addition, we denote the rank of publisher 1 when playing $\strategyProfile_1'$ by $\publisherRankLocation'$. Based on the utility function of the xMMR game,
    \begin{align*} 
        \publisherUtilityFunction_1(\strategyProfile_1, \strategyProfile_{-1}) & = \argmaxOn{\strategyProfile_1}{\aspectEShort{\simFunction(\strategyProfile_1, \aspectRep)}} >  \aspectEShort{\simFunction(\strategyProfile_1', \aspectRep)}
        \\ & \geq
        \aspectEShort{\simFunction(\strategyProfile_1', \aspectRep) \cdot \minOn{\rankIndex' < \publisherRankLocation'}{(\simFunction^\aspect_1 - \simLocationTag)^2}} = \publisherUtilityFunction_1(\strategyProfile_1', \strategyProfile_{-1}).
    \end{align*}
    
    Therefore, by definition, the \relSt of playing $\aspectEShort{\aspectRep}$ is a dominant strategy for publisher 1.
\end{proof}

\subsection{Lemma~\ref{lem:xquad uir unique pne}}\label{proof lem:xquad uir unique pne}
Before we prove the lemma, we will prove an auxiliary technical lemma, which generalizes Proposition~\ref{pro:sim maximizer}.
\begin{lemma}\label{lem:aux unique xquad expression}
    Let $\aspectDistribution$ be some distribution such that $\aspectDistributionSuppSize \geq 2$ and let $\partialRank$ be some partial ranking (can be $\emptyset$). If there is one or more aspects in the distribution support such that $\underset{\strategyTag \in \partialRank}{\Pi} \big(1 - \simPublisherTag\big) > 0$, then the \textbf{unique} maximizer of the UIR-xQuAD score function is a weighted average of those aspects' representations when the weights are $\aspectWeightAdjusted = \aspectProb \cdot \underset{\strategyTag \in \partialRank}{\Pi} \big(1 - \simPublisherTag\big)$.
\end{lemma}

\begin{proof}
    To prove the lemma, let us denote by $\baseWeightSetAdjusted_{real}$ the set of aspects $\aspectRep$ for which $\aspectWeightAdjusted > 0$. Our goal is to show that the unique maximizer of the UIR-xQuAD score is the weighted average of the aspects in $\baseWeightSetAdjusted_{real}$ with weights $\aspectWeightAdjusted$.

    \begin{align*}
        \scoreFunc\rb{\strategy, \aspectDistribution, \partialRank} &
        = \aspectEShort{\simPublisher \cdot \underset{\strategyTag \in \partialRank}{\Pi} \big(1 - \simPublisherTag\big)}
        \\ &
        = \sum_{\aspect \in \aspectDistributionSupp} \aspectProb \cdot \underset{\strategyTag \in \partialRank}{\Pi} \big(1 - \simPublisherTag\big) \cdot \simPublisher
        \\ &
        = \sum_{\aspect \in \baseWeightSetAdjusted_{real}} \aspectWeightAdjusted \cdot \simPublisher
        \\ &
        = \sum_{\aspect \in \baseWeightSetAdjusted_{real}} \aspectWeightAdjusted (1 - \frac{1}{\dimE}\norm{\strategy - \aspectRep}^2)
    \end{align*}

    Thus,
    \begin{align*}
        \argmaxOn{\strategy \in \embeddingSpaceBase}{\scoreFunc\rb{\strategy, \aspectDistribution, \partialRank}} &
        = \argmaxOn{\strategy \in \embeddingSpaceBase}{\sum_{\aspect \in \baseWeightSetAdjusted_{real}} \aspectWeightAdjusted (1 - \frac{1}{\dimE}\norm{\strategy - \aspectRep}^2)}
        \\ &
        = \argminOn{\strategy \in \embeddingSpaceBase}{\sum_{\aspect \in \baseWeightSetAdjusted_{real}} \aspectWeightAdjusted \norm{\strategy - \aspectRep}^2}
        \\ &
        = \frac{\sum_{\aspect \in \baseWeightSetAdjusted_{real}} \aspectWeightAdjusted \aspectRep}{\sum_{\aspect \in \baseWeightSetAdjusted_{real}} \aspectWeightAdjusted},
    \end{align*}
    where the last equality holds because the weighted squared distance is strictly convex in $\strategy$, ensuring a unique minimizer.
\end{proof}

We now proceed with the proof of the main lemma.

\begin{proof}
    Let $\game$ be a UIR-xQuAD game with $\aspectDistributionSuppSize \geq 2$.
    
    \paragraph{Part 1 - unique equilibrium}
    We will prove that $\game$ admits a unique Nash equilibrium.

    In this game, the user utility function is $\clickModel^{coverage}(\rankIndex, \simAllPublishersRank) = \simLocation \cdot \underset{\rankIndex' < \rankIndex}{\Pi} \big(1 - \simLocationTag\big)$. Notice that the user utility function is monotone as $0 \leq (1 - \simFunction(\strategy, \aspectRep)) \leq 1$. In addition, it is easy to see that $\clickModel(\rankIndex, \simAllPublishersRank)$ does not depend on $\simAllPublishersRank[\rankIndex']$ for $\rankIndex' > \rankIndex$. Therefore, the UIR score function conditions hold, and UIR-xQuAD is well-defined. What remains to show is that the maximizer of each expression $\scoreFunc\rb{\strategy, \aspectDistribution, \set{\strategyTag \given \publisherIndex' < \publisherIndex}}$ is unique.

    Suppose for contradiction that there is one or more publishers for which the maximizer of the expression is not unique.
    Let $\publisherIndex$ be the first index such that in the recursive formula, there is more than $1$ maximizer of the expression. For every $\publisherIndex' < \publisherIndex$, we will fix $\strategyTag$ to be the unique maximizer of the expression.
    In our case,
    \begin{equation*}
        \scoreFunc\rb{\strategy, \aspectDistribution, \set{\strategyTag \given \publisherIndex' < \publisherIndex}} = \aspectEShort{\simPublisher \cdot \underset{\publisherIndex' < \publisherIndex}{\Pi} \big(1 - \simPublisherTag\big)}.
    \end{equation*}

    We will now split into two cases base on $\cb{\underset{\publisherIndex' < \publisherIndex}{\Pi} \big(1 - \simPublisherTag\big)}_{\aspect \in \aspectDistributionSupp}$.

    \paragraph{Part 1, case 1:} 
    For every aspect $\aspect$, $\underset{\publisherIndex' < \publisherIndex}{\Pi} \big(1 - \simPublisherTag\big) = 0$.

    Notice that for $\simFunction(\strategyTag, \aspectRep) = 1 - \frac{1}{\dimE}\norm{\strategyTag - \aspectRep}^2$ to be $0$, $\strategyTag$ must be equal to $\aspectRep$. If $\underset{\publisherIndex' < \publisherIndex}{\Pi} \big(1 - \simPublisherTag\big)$ is $0$ for some aspect $\aspect$, then there must be some publisher $\publisherIndex'$ such that $\strategyTag = \aspectRep$. Notice that this means that $\aspectDistributionSuppSize < \publisherIndex$, and that the distribution is necessarily finite.
    
    We now focus on the \textbf{last} publisher (lexicographically) for which \textbf{two} of these products were positive, and denote her by $\publisherIndex'$. As $\aspectDistributionSuppSize \geq 2$, there must be such a publisher.
    From this publisher's perspective, there are $2$ aspects such that $\underset{\publisherIndex'' < \publisherIndex'}{\Pi} \big(1 - \simPublisherTagTag) > 0$. We will denote them by $\aspect_1, \aspect_2$. Using Lemma~\ref{lem:aux unique xquad expression}, we can discover that:    
    \begin{equation*}
        \frac{\aspectWeightAdjustedSpecific{1} \aspectRepOne + \aspectWeightAdjustedSpecific{2} \aspectRepTwo}{\aspectWeightAdjustedSpecific{1} + \aspectWeightAdjustedSpecific{2}}
        = \argmaxOn{\strategyTag}{\aspectEShort{\simPublisherTag \cdot \underset{\publisherIndex'' < \publisherIndex'}{\Pi} \big(1 - \simPublisherTagTag\big)}},
    \end{equation*} when $\aspectWeightAdjusted = \aspectProb \cdot \underset{\publisherIndex'' < \publisherIndex'}{\Pi} \big(1 - \simPublisherTagTag\big)$.

    Notice that the resulting $\strategyTag \notin \cb{\aspectRepOne, \aspectRepTwo}$, which contradicts the assumption that she is the \textbf{last} publisher (lexicographically) for which \textbf{two} of these products were positive. This means that the first case is not possible.

    \paragraph{Part 1, case 2:} At least one aspect satisfies
    $\underset{\publisherIndex' < \publisherIndex}{\Pi} \big(1 - \simPublisherTag\big) > 0$.
    
    By Lemma~\ref{lem:aux unique xquad expression}, the UIR-xQuAD score function expression for publisher $\publisherIndex$ has a unique maximizer, and we reached a contradiction. Therefore, the maximizer of each expression is unique.

    After we prove that each $\scoreFunc\rb{\strategy, \aspectDistribution, \set{\strategyTag \given \publisherIndex' < \publisherIndex}}$ expression has one maximizer, we can use the additional case of Theorem~\ref{the:direct pne}, which states that $\game$ has one unique equilibrium defined by the recursive formula.

    \paragraph{Part 2 - the symmetric case}
    We will show that if $\aspectDistribution$ is a symmetric distribution, then the \relSt of playing $\aspectEShort{\aspectRep}$ is a dominant strategy for each publisher. The proof of this part is similar to the proof of Theorem~\ref{the:xquad lex dominant}.
    
    Assume that $\aspectDistribution$ is symmetric. Fix some publisher $\publisherIndex \in \publishersSet$. Let $\strategiesOthers \in \prod_{\publisherIndex' \neq \publisherIndex} \embeddingSpaceSpecific{\publisherIndex'}$ be some strategies of all players except player $\publisherIndex$, $\strategy^1 = \aspectEShort{\aspectRep}$ and $\strategy^2 \neq \aspectEShort{\aspectRep}$.

    When we proved Lemma~\ref{lem:xquad lex relevance better rank}, we used a more general lemma, Lemma~\ref{lem:xquad general relevance better rank}, that allows for $\lambda \in \cb{0, 1}$. Notice that if we take the standard xQuAD score function, but with $\lambda = 0$ we get the UIR-xQuAD score function, which allows us to use the general lemma in our proof as well.
    
    According to Lemma~\ref{lem:xquad general relevance better rank}, every publisher positioned above publisher $\publisherIndex$ played $\aspectEShort{\aspectRep}$. In addition, according to this lemma, if publisher $\publisherIndex$ will deviate to $\strategy^2$, her position will not improve, as she can not be ranked above a publisher who plays $\aspectEShort{\aspectRep}$. We denote the rank of publisher $\publisherIndex$ when playing $\strategy^1$ by $\publisherRankLocation^1$ and the of publisher $\publisherIndex$ when playing $\strategy^2$ by $\publisherRankLocation^2$ and we get that $\publisherRankLocation^1 \leq \publisherRankLocation^2$.
    Therefore:
    \begin{align*}
        \publisherUtility(\strategy^2, \strategiesOthers) & = \aspectEShort{\simFunction(\strategy^2, \aspectRep) \cdot \underset{\rankIndex' < \publisherRankLocation^2}{\Pi} \big(1 - \simLocationTag\big)}
        \\ & \overset{(1)}{\leq} 
        \aspectEShort{\simFunction(\strategy^2, \aspectRep) \cdot \underset{\rankIndex' < \publisherRankLocation^1}{\Pi} \big(1 - \simLocationTag\big)}
        \\ & \overset{(2)}{<} 
        \aspectEShort{\simFunction(\strategy^1, \aspectRep) \cdot \underset{\rankIndex' < \publisherRankLocation^1}{\Pi} \big(1 - \simLocationTag\big)} = \publisherUtility(\strategy^1, \strategiesOthers)
    ,\end{align*}
    where transition~(1) is based on that $\publisherRankLocation^1 \leq \publisherRankLocation^2$ and that all the components are non-negative, and transition~(2) is Lemma~\ref{lem:xquad recursive maximizer}.

    As a result, by definition, the \relSt of playing $\aspectEShort{\aspectRep}$ is a dominant strategy for every publisher $\publisherIndex$.
    
    \paragraph{Part 3 - the non-symmetric case}
    We will show that if $\aspectDistribution$ is not a symmetric distribution and $\aspectDistributionSuppSize = 2$, then under the equilibrium, $\strategyProfile_1 \neq \strategyProfile_2$.

    We will denote the $2$ aspects by $\aspect_1$ and $\aspect_2$.
    For publisher 1, $\aspectWeightAdjusted = \aspectProb \cdot \underset{\strategyTag \in \emptyset}{\Pi} \big(1 - \simPublisherTag\big) = \aspectProb$
    and for publisher 2, $\aspectWeightAdjusted = \aspectProb \cdot \underset{\strategyTag \in \cb{\strategyProfile_1}}{\Pi} \big(1 - \simPublisherTag\big) = \aspectProb (1 - \simPublisherSpecific{1})$.
    
    Therefore, using Lemma~\ref{lem:aux unique xquad expression}, we can deduce that:
    \begin{gather*}
        \strategyProfile_1 = \aspectProbSpecific{1} \cdot \aspectRepOne + \aspectProbSpecific{2} \cdot \aspectRepTwo = \aspectEShort{\aspectRep}
        \\
        \strategyProfile_2 = \frac{\sum_{\aspect \in \cb{\aspect_1, \aspect_2}}\big(\aspectProb \norm{\aspectEShort{\aspectRep} - \aspectRep}^2 \big) \cdot \aspectRep}
        {\sum_{\aspect \in \cb{\aspect_1, \aspect_2}}\big(\aspectProb \norm{\aspectEShort{\aspectRep} - \aspectRep}^2 \big)}.
    \end{gather*}
    
    In order to show that $\strategyProfile_1 \neq \strategyProfile_2$, we just need to show the weights in both expressions are different, i.e.,
    \begin{equation*}
        \aspectProbSpecific{1} \neq \frac{\aspectProbSpecific{1} \norm{\aspectEShort{\aspectRep} - \aspectRepOne}^2}
        {\sum_{\aspect \in \cb{\aspect_1, \aspect_2}}\big(\aspectProb \norm{\aspectEShort{\aspectRep} - \aspectRep}^2 \big)}.
    \end{equation*}
    Suppose for contradiction that they are equal. Notice that both $\norm{\aspectEShort{\aspectRep} - \aspectRepOne}^2$ and $\norm{\aspectEShort{\aspectRep} - \aspectRepTwo}^2$ are positive.
    In this case:
    \begin{gather*}
        \aspectProbSpecific{1}
        =
        \frac{1}
        {1 + \frac{\aspectProbSpecific{2} \norm{\aspectEShort{\aspectRep} - \aspectRepTwo}^2}{\aspectProbSpecific{1} \norm{\aspectEShort{\aspectRep} - \aspectRepOne}^2}}
        \implies
        \\
        \aspectProbSpecific{1} + \frac{\aspectProbSpecific{2} \norm{\aspectEShort{\aspectRep} - \aspectRepTwo}^2}{ \norm{\aspectEShort{\aspectRep} - \aspectRepOne}^2}
        =
        1
    \end{gather*}

    We know that $\aspectProbSpecific{1} + \aspectProbSpecific{2} = 1$, so if the equivalence is true, $\norm{\aspectEShort{\aspectRep} - \aspectRepOne}^2$ must be equal to $\norm{\aspectEShort{\aspectRep} - \aspectRepTwo}^2$. By Auxiliary Lemma~\ref{lem:aux two aspects same weight}, the aspect distribution must be symmetric, and we reached a contradiction. Therefore, $\strategyProfile_1 \neq \strategyProfile_2$.
\end{proof}

\subsection{Lemma~\ref{lem:xmmr uir pne}}\label{proof lem:xmmr uir pne}
\begin{proof}
    Let $\game$ be a UIR-xMMR game.

    \paragraph{Part 1 - equilibrium existence}
    We start by showing that $\game$ admits at least one Nash equilibrium.
    In this game, the user utility function is $\clickModel^{novelty}(\rankIndex, \simAllPublishersRank) = \simLocation \cdot \minOn{\rankIndex' < \rankIndex}{{(\simLocation - \simLocationTag)}^2}$. Notice that the user utility function is monotone, as adding more elements to the minimum aggregation can only make it smaller. In addition, it is easy to see that $\clickModel(\rankIndex, \simAllPublishersRank)$ does not depend on $\simAllPublishersRank[\rankIndex']$ for $\rankIndex' > \rankIndex$. Therefore, the UIR score function conditions hold, and UIR-xMMR is well-defined. By Theorem~\ref{the:direct pne}, the strategy profile described by the recursive formula is a Nash equilibrium.

    \paragraph{Part 2 - strategies uniqueness}
    Suppose by contradiction that there is some Nash equilibrium profile $\strategyProfile$ and $2$ publishers $\publisherIndex \neq \publisherIndex'$ such that $\strategyProfile_\publisherIndex = \strategyProfile_{\publisherIndex'}$. Without loss of generality, assume that $\publisherIndex > \publisherIndex'$.
    
    By Auxiliary Lemma~\ref{lem:aux xmmr non zero strategy}, $\publisherUtility(\strategy, \strategiesOthers) = 0$ and there is some strategy $\strategyTilde$ such that $\publisherUtility^{xMMR}(\strategyTilde, \strategiesOthers) > 0$.
    Therefore, publisher $\publisherIndex$ has a profitable deviation, $\strategyProfile$ is not a Nash equilibrium, and we reached a contradiction.
\end{proof}

\section{Example}\label{app:example}

We now present a detailed explanation of the argument presented in Example~\ref{exa:xmmr no pne}, explicitly indicating the points where we used numerical methods implemented in the code.
During the example, we will use the definition of best response as presented in Definition~\ref{def:best response}.

Let $\game$ be an xMMR game with $\numPublishers = 3, \dimE = 1, \lambda = 0.5$ and the symmetric aspect distribution defined by $\aspectProb = 0.5 \cdot \ind(\aspectRep \in \cb{0, 0.1})$. Before we begin, note that the distribution we have presented satisfies the condition equivalent to symmetry (Auxiliary Lemma~\ref{lem:aux two aspects same weight}) and is therefore symmetric.

    Suppose for contradiction that there is some profile of $\strategyProfile = (\strategyProfile_1, \strategyProfile_2, \strategyProfile_3)$ which is NE.
    By Lemma~\ref{lem:xmmr first dominant}, publisher 1 has a dominant strategy of playing $\aspectEShort{\aspectRep}$ so $\strategyProfile_1$ must be $\aspectEShort{\aspectRep} = 0.05$. Notice that $\simFunction^{\aspect_1}_1 = \simFunction^{\aspect_1}_2 = 1 - 0.05^2 = \frac{399}{400} = 0.9975$.

    We start by finding the score function for the second-place comparison and the utility function of the second-ranked publisher:
    \begin{align*}
        \scoreFunc_{xMMR} & (\strategy; \aspectDistribution, \strategyProfile_1) =
        \frac{\simFunction^{\aspect_1}_\publisherIndex +  {(\simFunction^{\aspect_1}_\publisherIndex - \simFunction^{\aspect_1}_1)}^2}{4} + 
        \frac{\simFunction^{\aspect_2}_\publisherIndex + {(\simFunction^{\aspect_2}_\publisherIndex  - \simFunction^{\aspect_2}_1)}^2}{4}
        \\ & =
        \frac{1}{4} \Bigr( (1 - (\strategyProfile_{\rankLocationSpecific{2}})^2) + \big(1 - (\strategyProfile_{\rankLocationSpecific{2}})^2 - \frac{399}{400}\big)^2 \Bigr) + 
        \\ & \quad \quad +
        \frac{1}{4} \Bigr( (1 - (\strategyProfile_{\rankLocationSpecific{2}} - 0.1)^2) + \big(1 - (\strategyProfile_{\rankLocationSpecific{2}} - 0.1)^2 - \frac{399}{400}\big)^2 \Bigr)
        \\ & =
        \tfrac{1}{2}(\strategyProfile_{\rankLocationSpecific{2}})^4 - \tfrac{1}{10}(\strategyProfile_{\rankLocationSpecific{2}})^3 - \tfrac{39}{80}(\strategyProfile_{\rankLocationSpecific{2}})^2 + \tfrac{197}{4000}\strategyProfile_{\rankLocationSpecific{2}} + \tfrac{31841}{64000}
    ,\end{align*}
    \begin{align*}
        \publisherUtilityFunction_{\rankLocationSpecific{2}} & (\strategyProfile_{\rankLocationSpecific{2}}; \strategyProfile_1) = 
        \frac{\simFunction^{\aspect_1}_{\rankLocationSpecific{2}} \cdot {(\simFunction^{\aspect_1}_{\rankLocationSpecific{2}}  - \simFunction^{\aspect_1}_1)}^2}{2} + 
        \frac{\simFunction^{\aspect_2}_{\rankLocationSpecific{2}} \cdot {(\simFunction^{\aspect_2}_{\rankLocationSpecific{2}}  - \simFunction^{\aspect_2}_1)}^2}{2}
        \\ & =
        \frac{1}{2} (1 - (\strategyProfile_{\rankLocationSpecific{2}})^2)\big(1 - (\strategyProfile_{\rankLocationSpecific{2}})^2 - \frac{399}{400}\big)^2 
        \\ & \quad \quad +
        \frac{1}{2} (1 - (\strategyProfile_{\rankLocationSpecific{2}} - 0.1)^2)\big(1 - (\strategyProfile_{\rankLocationSpecific{2}} - 0.1)^2 - \frac{399}{400}\big)^2
        \\ & =
        -(\strategyProfile_{\rankLocationSpecific{2}})^6 + \tfrac{3}{10}(\strategyProfile_{\rankLocationSpecific{2}})^5 + \tfrac{93}{100}(\strategyProfile_{\rankLocationSpecific{2}})^4 - \tfrac{191}{1000}(\strategyProfile_{\rankLocationSpecific{2}})^3
        \\ & \quad \quad +
        \tfrac{3903}{160000}(\strategyProfile_{\rankLocationSpecific{2}})^2 - \tfrac{2367}{1600000}\strategyProfile_{\rankLocationSpecific{2}} + \tfrac{991}{32000000}
    .\end{align*}

    \begin{figure}
        \includegraphics[width=\columnwidth,keepaspectratio]{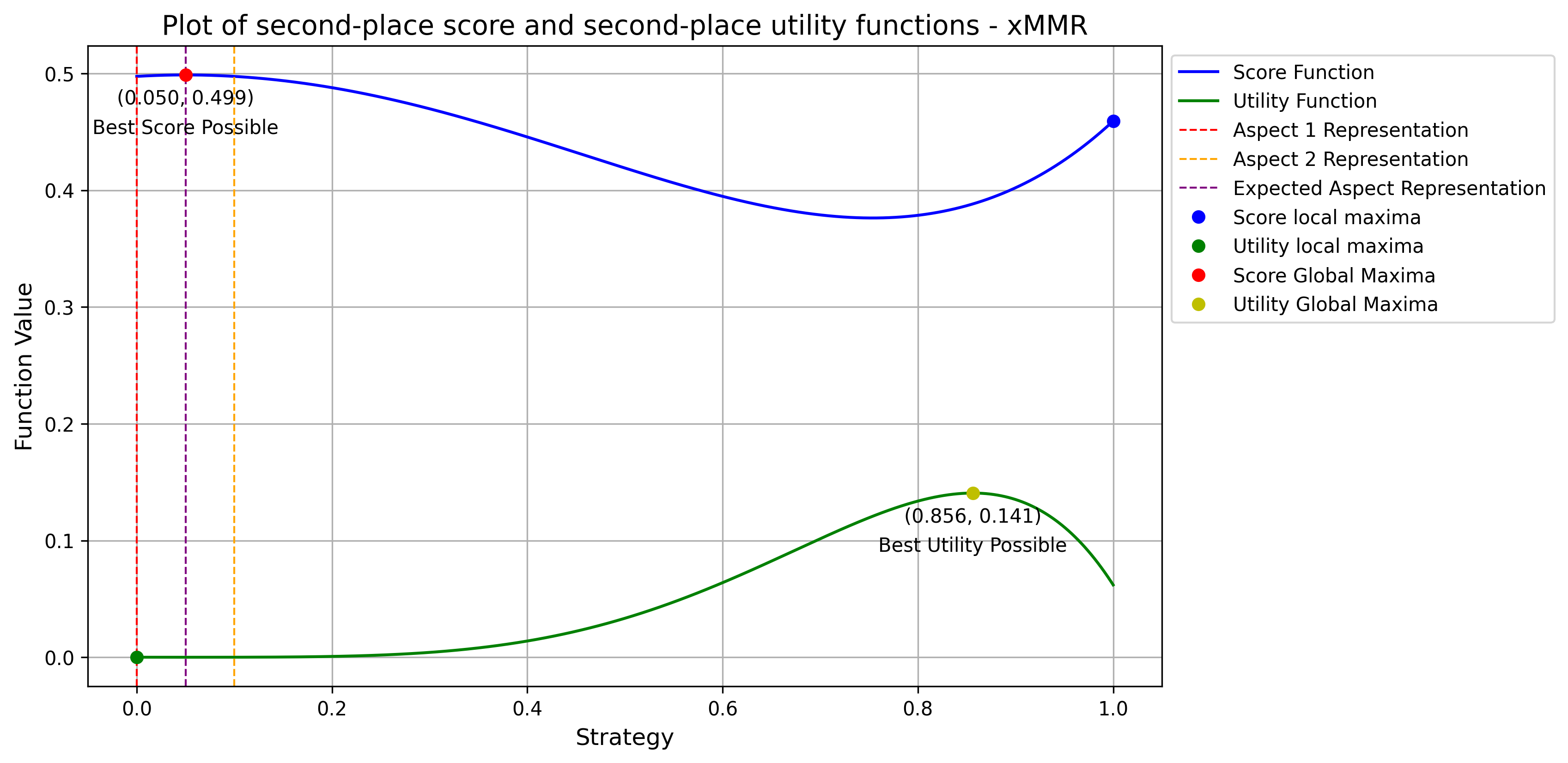}
        \caption{The score function for the second-place comparison and the utility function of the second-ranked publisher for the xQuAD game
        }\label{fig:xmmr second score utility}
        \Description{The figure shows two curves for the xMMR game: one is the score function for the second-place comparison, and another is the utility function of the second-ranked publisher. These functions illustrate how scores and utilities vary with the document representation.}
    \end{figure}

    Figure~\ref{fig:xmmr second score utility} presents the second position score function and the second-ranked publisher utility function, with their respective local and global maximum points. 
    We will denote the strategy that maximizes the utility function for the second-ranked publisher by $\strategyProfile_{\rankLocationSpecific{2}}^{\max}$.
    
    Now, we will split into cases based on the strategies of publishers 2 and 3. For all cases, we will use $\publisherIndex, \publisherIndex' \in \{2, 3 \}, \publisherIndex \neq \publisherIndex'$ to refer to the publishers when we do not want to point to them specifically.
    \begin{itemize}
        \item If $\strategyProfile_2 = \strategyProfile_3$: Publisher 3 will be ranked last according to the lexicographic tie-breaker, and because $\strategyProfile_2 = \strategyProfile_3$, she would get a utility of $0$. If she deviates and plays anything but $\aspectEShort{\aspectRep}$, she will get a positive utility, and that will be a profitable deviation.
        \item If $\publisherUtilityFunction_2(\strategyProfile) < \publisherUtilityFunction_3(\strategyProfile)$: Publisher 2 can improve her utility by playing $\strategyProfile_3$ as in this case she will be ranked second (due to the tie-breaker). If publisher 3 is currently ranked second, then by playing $\strategyProfile_3$ publisher 2 will get utility of $\publisherUtilityFunction_3(\strategyProfile)$, and if publisher 3 is currently ranked third, publisher 2 will get utility bigger or equal to $\publisherUtilityFunction_3(\strategyProfile)$, as the user utility function is monotone.
        \item If $\strategy = 0.05$: By Observation~\ref{obs:xmmr copycat zero} the utility of the publisher is $0$, and she has profitable deviation.
        \item If $\strategyProfile_2 = 1$: We showed using numerical methods that if publisher 3 plays the best response, she gets a utility bigger than the utility of publisher 2, which brings us to case 2. If publisher 3 does not play the best response, then she must have a profitable deviation, which is the best response. 
        \item If $\publisherUtilityFunction_2(\strategyProfile) > \publisherUtilityFunction_3(\strategyProfile)$ and $\strategyProfile_2 \notin \{ 0.05, 1 \}$: The current score of publisher 2 is not local maximum so for every $0 < \epsilon$, there is a strategy in $[\strategyProfile_2 - \epsilon, \strategyProfile_2 + \epsilon] \cap [0, 1]$ which results in a higher score value than the current score of publisher 2 (for the second position). Therefore, if we take $\epsilon \to 0$, for sufficiently small $\epsilon$, publisher 3 can deviate and get a higher score than publisher 2 while having a utility that converges to publisher 2's current utility (as the utility function is continuous if the rank is not changing), which must be a profitable deviation.
        \item If $\publisherUtilityFunction_2(\strategyProfile) = \publisherUtilityFunction_3(\strategyProfile)$, $\scoreFunc_{xMMR}(\strategy; \aspectDistribution, \strategyProfile_1) \neq \scoreFunc_{xMMR}(\strategyTag; \aspectDistribution, \strategyProfile_1)$: Let $\publisherIndex$ be the second-ranked publisher, and let $\publisherIndex'$ be the other publisher. We showed using numerical methods that if $\strategy = 0$ then $\strategyTag$ must be $0.1$ in order to have the same utility function value, but in this case the score values are also equal (which is because the distribution is symmetric), and this is the next case ($\scoreFunc_{xMMR}(\strategy; \aspectDistribution, \strategyProfile_1) = \scoreFunc_{xMMR}(\strategyTag; \aspectDistribution, \strategyProfile_1)$). In addition, $\strategy$ can not be $\strategyProfile_{\rankLocationSpecific{2}}^{\max}$ because no matter what publisher $\publisherIndex'$ plays, her utility will be lower. After ruling out $\strategy \in \{ 0, \strategyProfile_{\rankLocationSpecific{2}}^{\max} \}$, we can deduce that $\strategy$ is not at a local maximum of the utility function. Therefore, as $\scoreFunc_{xMMR}(\strategy; \aspectDistribution, \strategyProfile_1)$ is continuous, publisher $\publisherIndex$ can move a little step of size $\epsilon$ to increase her utility, while keeping her score greater than publisher $\publisherIndex'$ score, and this is a profitable deviation.
        \item If $\publisherUtilityFunction_2(\strategyProfile) = \publisherUtilityFunction_3(\strategyProfile)$, $\scoreFunc_{xMMR}(\strategy; \aspectDistribution, \strategyProfile_1) = \scoreFunc_{xMMR}(\strategyTag; \aspectDistribution, \strategyProfile_1)$, $\strategyProfile_2 \neq \strategyProfile_3$: After narrowing down the possibilities to this case, we used python to find pairs for which this condition hold, and all of the pairs founded are pairs of $[t, 0.1 - t]$, for $t \in [0, 0.1]$ which make sense as both the score function and the utility function of the second ranked publisher are symmetric with respect to $\aspectEShort{\aspectRep} = 0.05$. Let publisher $\publisherIndex$ be the publisher that played $\strategy > 0.05$, and let $\publisherIndex'$ be the other publisher for some pair in which the condition holds. We know that $\aspectRepOne = 0 \leq \strategyTag < \strategyProfile_1 < \strategy \leq 1 = \aspectRepTwo$, therefore $\abs{\simPublisher - \simFunction^{\aspect}_1} < \abs{\simPublisher - \simPublisherTag}$  for both aspects. Notice that the inequality is strong, so there is a step size $0 < \epsilon$ such that the inequality will also hold for $\strategy + \epsilon$. In this case, the utility function of publisher $\publisherIndex$ will be calculated based on publisher 1 similarity values, which means that although $\strategy + \epsilon$ will lead publisher $\publisherIndex$ to third-place, her utility will be like she was second (after publisher 1), therefore, this deviation will increase her utility, as shown in the second-place utility function graph and this is a profitable deviation.
    \end{itemize}

    We have shown that in every case there is a profitable deviation to at least one of the publishers and therefore we reached a contradiction and there is no Nash Equilibrium in game $\game$.

\end{document}